\documentclass[12pt]{article}
\usepackage{tikz}

\usepackage{amsfonts}
\usepackage{amsmath}
\usepackage{amsthm}
\usepackage{uhrzeit}
\usepackage{graphicx}
\usepackage{hyperref}
\usetikzlibrary{arrows}
\usetikzlibrary{calc}

\newtheorem{theorem}{Theorem}[section]
\newtheorem{lemma}[theorem]{Lemma}
\newtheorem{corollary}[theorem]{Corollary}

\definecolor{mygreen}{rgb}{0, 0.7, 0}%
\definecolor{myblue}{rgb}{0.2, 0.2, 1}%

\newcommand*{\N}{\ensuremath{\mathbb{N}}}
\DeclareMathOperator{\dist}{dist}
\def\qed{$\hfill\Box$}
\def\concorde{\texttt{Concorde}}

\title{\mbox{Hard to Solve Instances of the} \\ Euclidean Traveling Salesman Problem}
\author{Stefan Hougardy and Xianghui Zhong\\
\small Research Institute for Discrete Mathematics\\
\small University of Bonn\\
\small Lenn\'estr.~2, 53113 Bonn, Germany\\[5mm]
            }
\date{\today} 
          
\begin{document}
\maketitle 

\begin{abstract}
The well known $4/3$ conjecture states that the integrality ratio of the subtour LP
is at most $4/3$ for metric Traveling Salesman instances.  
We present a family of Euclidean Traveling Salesman instances for which we prove that the 
integrality ratio of the subtour LP converges to $4/3$. These instances (using the rounded Euclidean norm) 
turn out to be hard to solve exactly
with \concorde, the fastest existing exact TSP solver. 
For a 200 vertex instance from our family of Euclidean Traveling Salesman instances \concorde\ needs 
several days of CPU time. This is more than 1,000,000
times the runtime for a TSPLIB instance of similar size. 
Thus our new family of Euclidean Traveling Salesman instances may serve as new benchmark instances for TSP algorithms.
\end{abstract}

{\small\textbf{keywords:} traveling salesman problem; Euclidean TSP, integrality ratio, subtour LP, exact TSP solver}

\section{Introduction}
The traveling salesman problem (TSP) is probably the most well-known problem in discrete optimization. 
An instance is given by $n$ vertices and their pairwise distances. The task is to find a shortest tour 
visiting each vertex exactly once. This problem is known to be NP-hard~\cite{Kar1972}. 

If the distances satisfy the triangle inequality, we obtain an important special case 
called \textsc{Metric TSP}. For this problem, no better algorithm than the $\frac{3}{2}$-approximation algorithm 
proposed by Christofides in 1976~\cite{Chr1976}, which was independently developed by Serdjukov~\cite{serdjukov}, is known. 
A well studied special case of the \textsc{Metric TSP} is the \textsc{Euclidean TSP}. Here an instance consists 
of points in the Euclidean plane and distances are defined by the $l_2$ norm. 
The \textsc{Euclidean TSP} is still NP-hard~\cite{GGJ1976, Pap1977} but is in some sense easier than 
the \textsc{Metric TSP}: For the \textsc{Euclidean TSP} there exists a PTAS~\cite{Aro1998} while for the
\textsc{Metric TSP} there cannot exist a $\frac{123}{122}$-approximation algorithm unless $P=NP$~\cite{KLS2015}.

The \emph{subtour LP} is a relaxation of a well known integer linear program for the TSP~\cite{DFJ1954}. If $K_n$ is a complete graph with non-negative edge costs $c_e$ for all $e\in E(K_n)$ the subtour LP is given by:
\begin{align*}
\min \sum_{e\in E(K_n)} &c_ex_e\\
0 \leq x_e&\leq 1 && \text{for all~} e\in E(K_n)\\
\sum_{e \in \delta(v)} x_e &=2 && \text{for all~} v \in V(K_n) \\
\sum_{e \in E(K_n[X])}x_e &\leq |X|-1 && \text{for all~} X \subset V(K_n).
\end{align*}

Although this LP has exponentially many inequalities, the separation problem and hence the LP itself can be solved in polynomial time~\cite{GLS1981}. 

The \emph{integrality ratio} of a TSP instance is the ratio of the length of an optimum tour to the length of an optimum
solution to the subtour LP. The integrality ratio of a TSP variant is the supremum over the integrality ratios of all instances
of this TSP variant. 
The exact integrality ratio of the \textsc{Metric TSP} is not yet known but it must lie between 
$4/3$~\cite{Wil1990} and $3/2$~\cite{Wol1980}.
It is conjectured that the exact value is $4/3$~\cite[page 35]{Wil1990} and this conjecture is known under the name \emph{$4/3$-Conjecture}. 
For the \textsc{Euclidean TSP} we also know only that the integrality ratio must lie between $4/3$~\cite{Hou2014} 
and $3/2$~\cite{Wol1980}. 
The lower bound of $4/3$ was proven in~\cite{Hou2014} by showing that for a certain 
family of \textsc{Euclidean TSP} instances the integrality ratio converges to $4/3$. In these instances all vertices lie 
on three parallel lines whose distances depend on the number of vertices. \medskip

\noindent\textbf{New results. } 
In this paper we present a new family of instances of the \textsc{Euclidean TSP}
which we call \emph{tetrahedron instances} as they arise as certain subdivisions of a 2-dimensional projection of the edges of a tetrahedron.
For these tetrahedron instances we prove that the integrality ratio of the subtour LP converges to $4/3$.
The rate of convergence is faster than for the instances constructed in~\cite{Hou2014}. Moreover, 
knowing structurally different families of  instances for the \textsc{Metric TSP} with integrality ratio converging to $4/3$
may be useful for attacking the $4/3$-Conjecture. 
Finding optimum solutions for the tetrahedron instances turns out to be much more difficult than for any known metric TSP instances of similar sizes:
When using \concorde~\cite{ABCC2003b}, the fastest known exact TSP solver, we observe that on instances with about 200 vertices, 
\concorde\ is more than 1,000,000 times slower than on TSPLIB instances~\cite{Rei1995} of similar size. 
Therefore, our tetrahedron instances may serve as new benchmark instances for the TSP 
and we provide them for download in TSPLIB format~\cite{Tnmurl}. \medskip

\noindent\textbf{Outline of the paper.} 
In Section~\ref{sec:TetrahedronInstances} we present the construction of the tetrahedron instances and
introduce a modification of the tetrahedron instances which results in instances with (up to symmetry) unique optimum tours. 
In Section~\ref{sec:structure} we prove some structural results 
of the optimum tours in these modified tetrahedron instances which allow us
to bound the length of an optimum TSP tour in these instances. We also compute a bound for an optimum solution of the subtour LP
for the modified tetrahedron instances. By combining these two results we can prove that 
for certain families of the modified tetrahedron instances the integrality ratio 
converges to $4/3$. We then show how to carry over this result to the (unmodified) tetrahedron instances.

In Section~\ref{sec:Experiments} we present runtime experiments with \concorde~\cite{ABCC2003b} on the tetrahedron instances. 
We compare the runtimes with the runtimes on the instances proposed in~\cite{Hou2014} and  
on instances of comparable size from the TSPLIB~\cite{Rei1995}.

\section{The Tetrahedron Instances and Their Structural Properties} \label{sec:TetrahedronInstances}

In this section we first define the tetrahedron instances and the modified tetrahedron instances. 
Then, we prove some general properties of the optimal tour and geometrical properties of these instances. 
With this preparation we show structure theorems for the optimal tour that determine it uniquely up to certain symmetries.

\subsection{Construction of the Tetrahedron Instances \boldmath $T_{n,m}$}
Now, we construct the instance $T_{n,m}$. Figure~\ref{fig:T(9,5)} shows as an example the instance $T_{9,5}$. 
Denote the Euclidean distance between two points $x$ and $y$ in the plane by $\dist(x,y)$.
Let $A,B,C$ be the vertices 
of an equilateral triangle with center $M$. The three sides of the triangle are called \emph{base sides}, 
the closed line segments connecting $A$, $B$, or $C$ to $M$ are called \emph{internal segments}. 
The vertex $M$ belongs to all three internal segments. 
Denote the base side opposite to $A,B$ respectively $C$ by $a,b$ respectively $c$ and the internal segments connecting $A,B$ 
respectively $C$ with $M$ by $e,f$ respectively $g$. 

Given such an equilateral triangle $ABC$ we define the \emph{tetrahedron instance} $T_{n,m}$ for $n, m \in\N$ as follows.
We refine each of the base sides $a$, $b$, and $c$ by $n-1$ equidistant vertices 
$a_1,\dots,a_{n-1}$, $b_1,\dots,b_{n-1}$, and $c_1,\dots,c_{n-1}$. Moreover we define
$a_0 := B$, $a_n := C$, $b_0 := C$, $b_n := A$, $c_0 := A$, and $c_n := B$.
For $i \in  \{0, \ldots, n\}$ the vertices $a_i$, $b_i$, and $c_i$ are called \emph{base vertices}.

Similarly, we refine each of the internal segments $e$, $f$, and $g$ by $m - 1$ equidistant vertices 
$e_1,\dots,e_{m-1}$, $f_1,\dots,f_{m-1}$, and $g_1,\dots,g_{m-1}$ numbered in ascending order from 
$A =: e_0$, $B =: f_0$,  respectively $C =: g_0$ to the center $M =:e_m = f_m = g_m$. 
For $i \in \{ 1,\dots, m\}$ the vertices $e_i$, $f_i$, and  $g_i$ are called \emph{internal vertices}. 

Finally, we rotate and scale all coordinates such that the side $c$ is parallel to the $x$-axis, the vertex $C$ is above
the side $c$, and 
the distance between two consecutive base vertices is~$1$:
\begin{equation}
\dist(a_i,a_{i+1})=\dist(b_i,b_{i+1})=\dist(c_i,c_{i+1})= 1  \quad \mbox{for} \ 0\leq i < n.
\end{equation}
This implies $\displaystyle \dist(A,M) = \dist(B,M) = \dist(C, M) = \frac{n}{\sqrt 3}$ and therefore
\begin{equation}
\dist(e_i,e_{i+1})=\dist(f_i,f_{i+1})=\dist(g_i,g_{i+1})= \frac{n}{\sqrt 3\cdot m} \quad \mbox{for} \ 0\leq i < m \label{dist_e_i_e_i+1}.
\end{equation}
The smallest possible distance between any two different internal vertices will be denoted by $\gamma$, i.e., we have
\begin{equation}
\gamma = \frac{n}{\sqrt 3\cdot m} . \label{gamma}
\end{equation}

In total, the instance $T_{n,m}$ has $3(n+m) - 2$ vertices and a possible way to assign explicit coordinates to all 
these vertices satisfying the above conditions is to assign for $i = 0, \ldots, n$ and for $j=0, \ldots, m$:

\begin{equation}\left.~~~~~
\begin{array}{rcrr}
 a_i & ~:=~ \big( &   n - i/2         ,~ & i\cdot \sqrt 3 /2 \big) \\[2mm]
 b_i & ~:=~ \big( & n/2 - i/2         ,~ & (n-i)\cdot \sqrt 3 / 2 \big) \\[2mm]
 c_i & ~:=~ \big( &        i          ,~ &  0 \vphantom{\sqrt 3 /2} \big) \\[2mm]
 e_j & ~:=~ \big( & j \cdot n/(2m)    ,~ & j \cdot n/(2 \sqrt 3 m) \big) \\[2mm]
 f_j & ~:=~ \big( & n - j \cdot n/(2m),~ & j \cdot n/(2 \sqrt 3 m)\big) \\[2mm]
 g_j & ~:=~ \big( & n/2               ,~ & n \cdot \sqrt 3 / 2 - j \cdot n /(\sqrt 3 m)\big) \\[2mm]
\end{array} ~~~~~~~~~~\right\}
\end{equation}

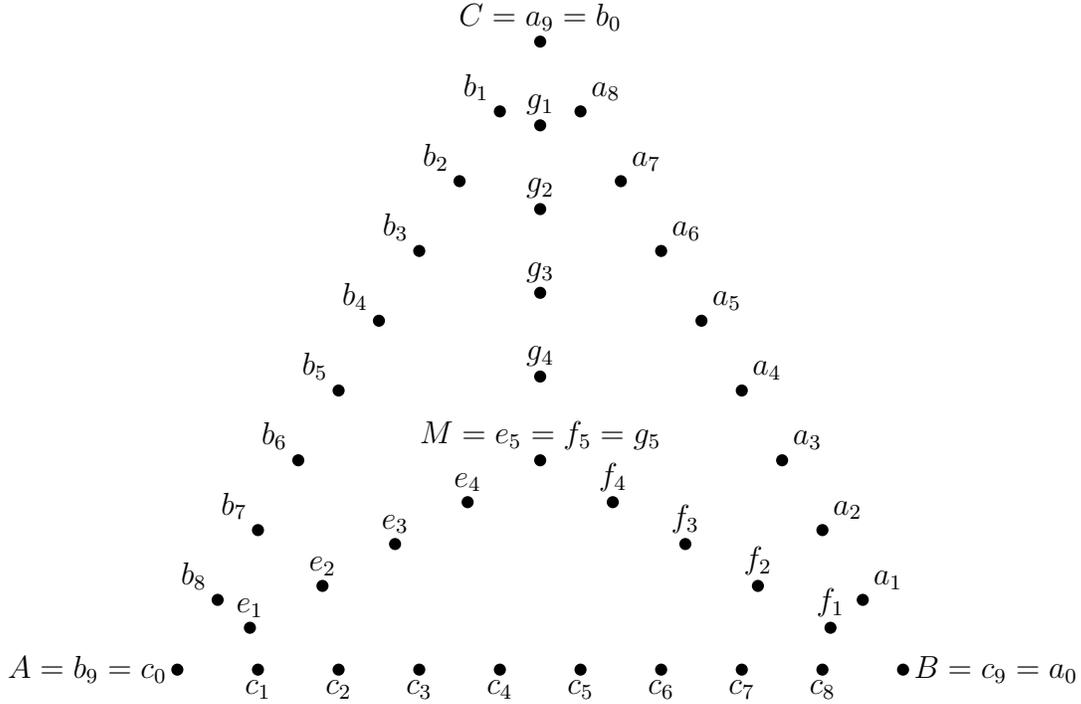
\begin{figure}[ht]
\centering
\begin{tikzpicture}[scale=1.072]
\def\n{9}
\def\m{5}
\def\sqrt3{1.73205}
\def\mycircle{circle(0.75mm)}

\pgfmathparse{\n - 1}
\foreach \i in {1,...,\pgfmathresult}
{
   \fill (\n     - \i / 2, \i * \sqrt3 / 2) \mycircle node[anchor = south west] {$a_{\i}$};
   \pgfmathparse{(\n - \i) * \sqrt3 / 2}
   \fill (\n / 2 - \i / 2,  \pgfmathresult) \mycircle  node[anchor = south east] {$b_{\i}$};
   \fill (             \i,               0) \mycircle  node[anchor = north] {$c_{\i}$};
}

\pgfmathparse{\m - 1}
\foreach \j in {1,...,\pgfmathresult}
{
   \fill(     \j * \n / 2 / \m, \j * \n / 2 / \sqrt3 / \m) \mycircle node[anchor = south] {$e_{\j}$};
   \fill(\n - \j * \n / 2 / \m, \j * \n / 2 / \sqrt3 / \m) \mycircle node[anchor = south] {$f_{\j}$};
   \fill(               \n / 2, \n * \sqrt3 / 2 - \j * \n / \sqrt3 / \m) \mycircle node[anchor = south] {$g_{\j}$};
}

\fill(\n / 2 ,  \n / 2 / \sqrt3) \mycircle node[anchor = south] {$M = e_\m = f_\m = g_\m$};
\fill(\n, 0) \mycircle node[anchor = west] {$B = c_\n = a_0$};
\fill(\n / 2,  \n * \sqrt3 / 2) \mycircle node[anchor = south] {$C = a_\n = b_0$};
\fill(0, 0) \mycircle node[anchor = east] {$A = b_\n = c_0$};
\end{tikzpicture}
\caption{The tetrahedron instance $T_{9,5}$.}
\label{fig:T(9,5)}
\end{figure}

Whenever we use the words left, right, above or below to express the relative position between two points
in an instance $T_{n,m}$, we always consider this instance to be embedded in such a way that $c$ is parallel to the $x$-axis and $A$, $B$, and $C$ are oriented counterclockwise.

\subsection{The Modified Tetrahedron Instances \boldmath $T'_{n,m}$}

To analyze the integrality ratio of the instances $T_{n,m}$ it turns out to be useful to introduce slightly
modified instances $T'_{n,m}$. The instance $T'_{n,m}$ is obtained from $T_{n,m}$ by removing all internal vertices 
whose distance to one of the vertices $A$, $B$, or $C$ is less than $\max\{10, 4 + 4 \gamma\}$. Figure~\ref{fig:T'(48,24)} shows the instance $T'_{48,24}$.
We have: 
\begin{equation}
\dist(A, e_j) \ge \max\{10, 4 + 4 \gamma\} \mbox{ for all } e_j\in T'_{n,m}. \label{mindistproperty}
\end{equation}
As $\dist(c_i, e_j) \ge \dist(A, e_j) \cdot \sin (30^\circ) = \frac12 \cdot \dist(A, e_j)$ we get 
\begin{equation}
\dist(c_i, e_j) \ge 5 \mbox{ for all } c_i, e_j\in T'_{n,m}. \label{mindistproperty2}
\end{equation}
We will need that an instance $T'_{n,m}$ contains at least four internal vertices. We therefore 
require that 
\begin{equation}
n \ge 40  \mbox{ ~and~ } m \ge 22.  \label{eq:nm}
\end{equation}
Using~(\ref{dist_e_i_e_i+1}) this implies 
$$ \dist(A, e_{m-1}) ~=~ \frac{(m-1)\cdot n}{\sqrt 3 \cdot m} ~\ge~ \frac {21}{22 \sqrt 3}\cdot n ~>~ \frac{n}{2} ~\ge~ 10$$
and using~(\ref{gamma}) this implies 
$$ \dist(A, e_{m-1}) ~=~ \frac{(m-5)\cdot n}{\sqrt 3 \cdot m} + \frac{4\cdot n}{\sqrt 3 \cdot m} \ge \frac {17}{22 \sqrt 3}\cdot n + 4 \gamma ~>~ 4+ 4 \gamma$$
and therefore $T'_{n,m}$ contains at least the four internal vertices $e_{m-1}$, $f_{m-1}$, $g_{m-1}$, and $M$ if 
$n \ge 40$  and  $m \ge 22$. 

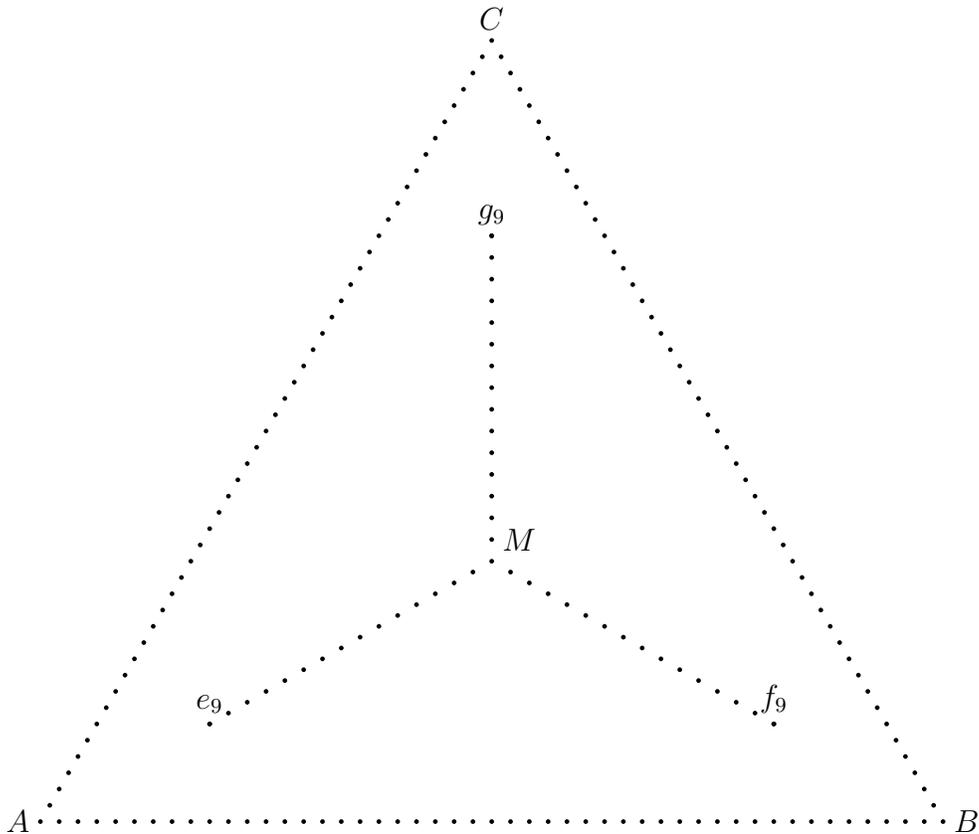
\begin{figure}[ht]
\centering
\begin{tikzpicture}[scale=0.25]

\def\n{48}
\def\m{24}
\def\sqrt3{1.73205}
\def\mycircle{circle(1.25mm)}

\pgfmathparse{(\n - 1)}
\foreach \i in {1,...,\pgfmathresult}
{
   \fill (\n     - \i / 2, \i * \sqrt3 / 2) \mycircle ;
   \pgfmathparse{(\n - \i) * \sqrt3 / 2}
   \fill (\n / 2 - \i / 2,  \pgfmathresult) \mycircle;
   \fill (             \i,               0) \mycircle;
}

\pgfmathparse{\m - 1}
\foreach \j in {9,...,\pgfmathresult}
{
   \fill(     \j * \n / 2 / \m, \j * \n / 2 / \sqrt3 / \m) \mycircle ;
   \fill(\n - \j * \n / 2 / \m, \j * \n / 2 / \sqrt3 / \m) \mycircle ;
   \fill(               \n / 2, \n * \sqrt3 / 2 - \j * \n / \sqrt3 / \m) \mycircle ;
}

\pgfmathparse{\m - 1}
\foreach \j in {9}
{
   \fill(     \j * \n / 2 / \m, \j * \n / 2 / \sqrt3 / \m) \mycircle node[anchor = south] {$e_{\j}$};
   \fill(\n - \j * \n / 2 / \m, \j * \n / 2 / \sqrt3 / \m) \mycircle node[anchor = south] {$f_{\j}$};
   \fill(               \n / 2, \n * \sqrt3 / 2 - \j * \n / \sqrt3 / \m) \mycircle node[anchor = south] {$g_{\j}$};
}

\fill(\n / 2 ,  \n / 2 / \sqrt3) \mycircle node[anchor = south west] {$M$};
\fill(\n, 0) \mycircle node[anchor = west] {$B$};
\fill(\n / 2,  \n * \sqrt3 / 2) \mycircle node[anchor = south] {$C$};
\fill(0, 0) \mycircle node[anchor = east] {$A$};
\end{tikzpicture}
\caption{The modified tetrahedron instance $T'_{48,24}$.}
\label{fig:T'(48,24)}
\end{figure}

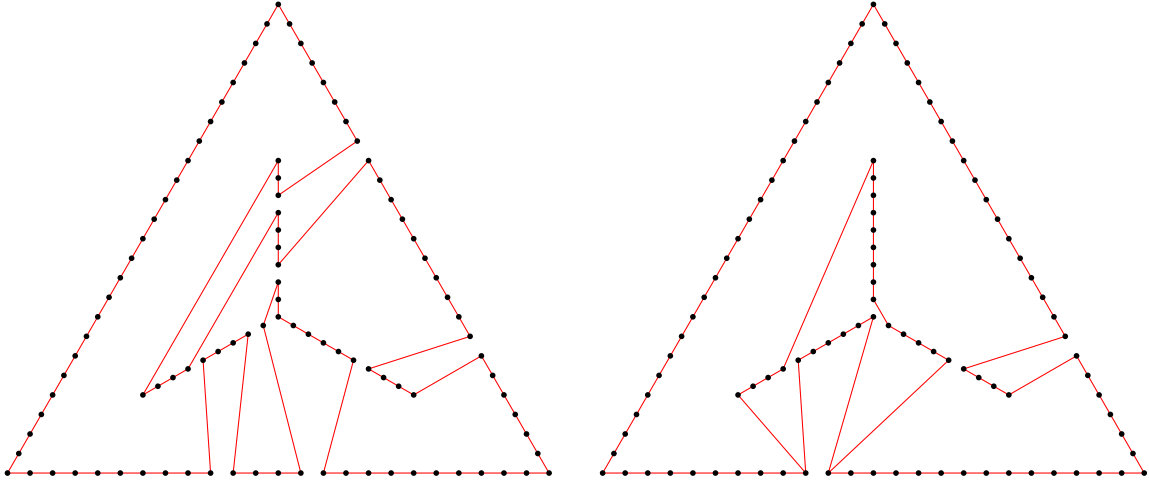
\begin{figure}[ht]
\centering

\begin{tikzpicture}[scale=0.3]
\def\n{24}
\def\m{18}
\def\sqrt3{1.73205}
\def\mycircle{circle(1.25mm)}

\def\a#1{(\n     - #1 / 2, #1 * \sqrt3 / 2)}
\def\c#1{(#1, 0)}
\def\e#1{(#1 * \n / 2 / \m, #1 * \n / 2 / \sqrt3 / \m)}
\def\f#1{(\n - #1 * \n / 2 / \m, #1 * \n / 2 / \sqrt3 / \m)}
\def\g#1{(\n / 2, \n * \sqrt3 / 2 - #1 * \n / \sqrt3 / \m)}

\draw[red]  (0, 0) -- (\n / 2,  \n * \sqrt3 / 2) ;
\draw[red]  \c0 -- \c9;
\draw[red]  \c{10} -- \c{13};
\draw[red]  \c{14} -- \a0;
\draw[red]  \a{16} -- \g{15} -- \g{12} -- \e{12} -- \e9 -- \g9 -- \g{11} -- \a{17} -- \a{24};
\draw[red]  \a0 -- \a6 -- \f9 -- \f{12} -- \a7 -- \a{16};
\draw[red]  \c9 -- \e{13} -- \e{16} -- \c{10};
\draw[red]  \c{13} -- \e{17} -- \g{16} -- \g{18} -- \f{13} -- \c{14};

\pgfmathparse{(\n - 1)}
\foreach \i in {1,...,\pgfmathresult}
{
   \fill (\n     - \i / 2, \i * \sqrt3 / 2) \mycircle ;
   \pgfmathparse{(\n - \i) * \sqrt3 / 2}
   \fill (\n / 2 - \i / 2,  \pgfmathresult) \mycircle;
   \fill (             \i,               0) \mycircle;
}

\pgfmathparse{\m - 1}
\foreach \j in {9,...,\pgfmathresult}
{
   \fill(     \j * \n / 2 / \m, \j * \n / 2 / \sqrt3 / \m) \mycircle ;
   \fill(\n - \j * \n / 2 / \m, \j * \n / 2 / \sqrt3 / \m) \mycircle ;
   \fill(               \n / 2, \n * \sqrt3 / 2 - \j * \n / \sqrt3 / \m) \mycircle ;
}

\fill(\n / 2 ,  \n / 2 / \sqrt3) \mycircle ;
\fill(\n, 0) \mycircle ;
\fill(\n / 2,  \n * \sqrt3 / 2) \mycircle ;
\fill(0, 0) \mycircle ;

\end{tikzpicture}\hspace*{5mm}
\begin{tikzpicture}[scale=0.3]
\def\n{24}
\def\m{18}
\def\sqrt3{1.73205}
\def\mycircle{circle(1.25mm)}

\def\a#1{(\n     - #1 / 2, #1 * \sqrt3 / 2)}
\def\c#1{(#1, 0)}
\def\e#1{(#1 * \n / 2 / \m, #1 * \n / 2 / \sqrt3 / \m)}
\def\f#1{(\n - #1 * \n / 2 / \m, #1 * \n / 2 / \sqrt3 / \m)}
\def\g#1{(\n / 2, \n * \sqrt3 / 2 - #1 * \n / \sqrt3 / \m)}

\draw[red]  (0, 0) -- (\n / 2,  \n * \sqrt3 / 2) ;
\draw[red]  \c0 -- \c9;
\draw[red]  \c{10} -- \a0;

\draw[red]  \a0 -- \a6 -- \f9 -- \f{12} -- \a7 -- \a{24};
\draw[red]  \c9 -- \e{13} -- \e{18} -- \c{10};
\draw[red]  \c{9} -- \e{9} -- \e{12} -- \g9 -- \g{17} -- \f{17} -- \f{13} -- \c{10};

\pgfmathparse{(\n - 1)}
\foreach \i in {1,...,\pgfmathresult}
{
   \fill (\n     - \i / 2, \i * \sqrt3 / 2) \mycircle ;
   \pgfmathparse{(\n - \i) * \sqrt3 / 2}
   \fill (\n / 2 - \i / 2,  \pgfmathresult) \mycircle;
   \fill (             \i,               0) \mycircle;
}

\pgfmathparse{\m - 1}
\foreach \j in {9,...,\pgfmathresult}
{
   \fill(     \j * \n / 2 / \m, \j * \n / 2 / \sqrt3 / \m) \mycircle ;
   \fill(\n - \j * \n / 2 / \m, \j * \n / 2 / \sqrt3 / \m) \mycircle ;
   \fill(               \n / 2, \n * \sqrt3 / 2 - \j * \n / \sqrt3 / \m) \mycircle ;
}

\fill(\n / 2 ,  \n / 2 / \sqrt3) \mycircle ;
\fill(\n, 0) \mycircle ;
\fill(\n / 2,  \n * \sqrt3 / 2) \mycircle ;
\fill(0, 0) \mycircle ;

\end{tikzpicture}

\caption{A tour with four trips in the instance  $T'_{24,18}$ (left). A pseudo-tour in the instance $T'_{24,18}$ (right).}
\label{fig:T'(24,18)}
\end{figure}

Note that as long as the ratio $n/m$ is constant, the instances $T_{n,m}$ and $T'_{n,m}$ differ by a constant number of vertices only. 
Therefore, as we will see later, $T_{n,m}$ and $T'_{n,m}$ have asymptotically the same integrality ratio, if $n/m$ is constant.

\section{Structural Properties of Optimal Tours}\label{sec:structure}
In this section we show that the shape of the tour shown in Figure~\ref{fig:OptimumTour} is the unique optimal tour for $T'_{n,m}$ up to symmetry. To achieve this, we first summarize some previous results for optimal tours. Then, we introduce the concept of a pseudo-tour from which we derive structural properties of optimal tours in Theorem~\ref{long}. After that, we make further geometrical observations culminating in Theorem~\ref{thm:uniqueOpt} describing the unique optimal tour up to symmetry.

A \emph{tour} for a Euclidean Traveling Salesman instance is a polygon that contains all the given vertices.
A polygon is called \emph{simple} if no two of its edges intersect except for the common vertex of two consecutive edges.
As an immediate consequence from the triangle inequality we have

\begin{lemma}[Flood 1956~\cite{Flo1956}] \label{lemma:nocrossing}
Unless all vertices lie on one line, an optimal tour of a Euclidean Traveling Salesman instance is a simple polygon. \qed
\end{lemma}

An important consequence of Lemma~\ref{lemma:nocrossing} is the following result:

\begin{lemma}[\cite{DDR1994}, page 142] \label{lemma:convexhull}
An optimal tour of a Euclidean Traveling Salesman instance visits the vertices on the boundary of the convex hull 
of all vertices in their cyclic order. \qed
\end{lemma}

In case of the tetrahedron instances, the vertices on the boundary of the convex hull are
exactly the base vertices. 
From now on we fix the orientation of an optimal TSP tour of the tetrahedron instance 
such that the base vertices are visited in counterclockwise order.
We use the notation $(x, y)$ for an edge of an oriented optimum tour.

A \emph{subpath} of an oriented tour consists of vertices $v_1,\dots,v_l$, such that $v_{i+1}$ is visited 
by the tour immediately after $v_i$ for all $i=1,\dots, l-1$. 
A subpath of a tour starting and ending at base vertices and containing no other base vertices is called a \emph{trip} 
if it contains at least one internal vertex. The first and the last internal vertex of a trip (which may coincide)
are called the \emph{connection vertices} of the trip.
By Lemma~\ref{lemma:convexhull} the base vertices are visited counterclockwise. Therefore, the two end vertices of a trip
must be consecutive base vertices belonging to the same side of the triangle $ABC$. We call this side the
\emph{main side} of the trip.

Each optimum tour of a tetrahedron instance can be decomposed into a set of trips and a set
of edges connecting consecutive base vertices such that all internal vertices are contained in some trip and two different trips
intersect in at most one base vertex (see Figure~\ref{fig:T'(24,18)} left). A set of edges between consecutive base vertices together with a set of trips 
that contain all internal vertices is called
a \emph{pseudo-tour} if it is not a tour and satisfies the following property:  for any two consecutive base vertices  there exists at 
least one trip 
having these two vertices as end vertices if and only if these two vertices are not connected by an edge (see Figure~\ref{fig:T'(24,18)} right). \bigskip

The following result will be useful to prove that certain tours of the tetrahedron instances are not optimum.

\begin{lemma} A pseudo-tour in $T'_{n,m}$ is more than 1 longer than an optimum tour in $T'_{n,m}$.
\label{lemma:pseudo-tour}
\end{lemma}
\begin{proof}
Let $T$ be a pseudo-tour of minimum length in $T'_{n,m}$. There exist
two consecutive base vertices that are the end vertices of at least two trips. Wlog we may assume that the two base vertices
are $c_i$ and $c_{i+1}$. Let $c_i, x_1, \ldots, x_k, c_{i+1}$ be the vertices of the first trip and  
$c_i, y_1, \ldots, y_l, c_{i+1}$ be the vertices of a second trip with $k, l\ge 1$. By~(\ref{mindistproperty2}) we know that 
the distance from $c_i$ to $x_1$ or $y_1$ and from $c_{i+1}$ to $x_k$ or $y_l$ is at least~$5$.
Moreover, we have $\dist(c_i, c_{i+1}) = 1$ and $T$ is intersection free. If there is no $u\in\{x_k,y_l\}$ and $v\in\{x_1,y_1\}$ such that the rays $\overrightarrow{c_iu}$ and $\overrightarrow{c_{i+1}v}$ intersect, we have $\angle x_1c_iy_1+\angle x_kc_{i+1}y_l \leq 180^\circ$. Otherwise, consider the intersection $P$ of two of these rays such that no other intersection point lies on or in the triangle $c_ic_{i+1}P$ and let $Q$ be the intersection of the angle bisector of $\angle c_iPc_{i+1}$ with $c$. Then,
\begin{align*}
\angle x_1c_iy_1+\angle x_kc_{i+1}y_l&\leq 180^\circ-\angle c_{i+1}c_iP+180^\circ-\angle Pc_{i+1}c_i=360^\circ-(180^\circ-\angle c_iPc_{i+1})\\
&=180^\circ+\angle c_iPc_{i+1}.
\end{align*}
On the other hand either $\dist(c_i,Q)\leq \frac{1}{2}$ or $\dist(Q,c_{i+1})\leq \frac{1}{2}$ so wlog $\dist(c_i,Q)\leq \frac{1}{2}$.
Using the law of sines we get:

\begin{align*}
\sin(\frac{\angle c_iPc_{i+1}}{2})=\sin(\angle c_iPQ)=\frac{\dist(c_i,Q)}{\dist(c_i,P)}\sin(\angle PQc_i)\leq \frac{\frac{1}{2}}{5}<\frac{1}{5}.
\end{align*}
Therefore, the total sum of the two angles $x_1 c_i y_1$ and $x_k c_{i+1} y_l$ is at most $180^\circ + 2 \alpha$
with $\sin \alpha < 1/5$. 

Wlog we may assume that  the angle $x_1 c_i y_1$ is at most  $90^\circ + \alpha$. 
Set $x := \dist(c_i, x_1)$, $y := \dist(c_i, y_1)$, and $z := \dist(x_1, y_1)$. As by~(\ref{mindistproperty2}) $x, y \ge 5$  we get
$-4x + xy \ge xy/5$ and $-4y + xy \ge xy/5$ which implies
$$4 - 4 x - 4 y + 2 x y > 2 x y /5.$$
By adding $x^2 + y^2$ to both sides of this inequality and using $\frac{1}{5} > \sin(\alpha) = -\cos(90^\circ + \alpha)$ we get
with the law of cosines
$$(x+y-2)^2 > x^2 + y^2 + 2xy/5 > x^2 + y^2 - 2 x y \cos(90^\circ + \alpha) \ge z^2.$$
Therefore, we have $$z + 2 < x + y.$$
We now replace the edges $(c_i, x_1)$ and $(c_i, y_1)$ by the edges $(c_i, c_{i+1})$
and $(x_1, y_1)$ and shortcut the two edges $(c_i, c_{i+1})$ and $(c_{i+1}, y_l)$ by the edge $(c_i, y_l)$.
This yields either a pseudo-tour that is shorter than $T$, contradicting the choice of $T$. Or it yields a tour
that is by more than~$1$ shorter than $T$ which proves the result. 
\end{proof}

\begin{lemma} 
In an optimum tour for $T'_{n,m}$ a trip with end vertices $c_i$ and $c_{i+1}$ cannot contain the edge $(c_i, g_j)$ for 
$1 \le j \le m-1$. 
\label{lemma:nogi}
\end{lemma}

\begin{proof}
Suppose there exists a trip with end vertices $c_i$ and $c_{i+1}$ that contains the edge $(c_i, g_j)$ for some $j\in \{1,\ldots, m-1\}$.
Let $(x, c_{i+1})$ be the last edge of the trip. 

If $x \in \{f_1, \ldots, f_m, g_1, \ldots, g_{m-1}\}$ then we can replace
the edges $(c_i, g_j)$ and $(x, c_{i+1})$ by $(a_{n-i}, g_j)$ and $(x, a_{n-(i+1)})$ and add the edge $(c_i, c_{i+1})$.
Note that $\dist(c_i, g_j) > \dist(a_{n-i}, g_j)$ and $\dist(x, c_{i+1}) \ge \dist(x, a_{n-(i+1)})$. 
If  $(a_{n-(i+1)}, a_{n-i})$ is an edge in the tour, then we can remove this edge and get a shorter tour, contradiction. 
Otherwise, we got a pseudo-tour that is at most~1 longer than an optimum tour. This contradicts Lemma~\ref{lemma:pseudo-tour}.

If $x \in \{e_1, \ldots, e_{m-1}\}$ we can replace
the edges $(c_i, g_j)$ and $(x, c_{i+1})$ by $(b_{n-i}, g_j)$ and $(x, b_{n-(i+1)})$ and use an analogous argument.
\end{proof}

From now on we denote by $x'$ the reflection of $x$ across $c$.

\begin{lemma} \label{firstlast}
Let $(c_i,x)$ and $(y,c_{i+1})$ be the first and last edge of a trip in an optimum tour for $T_{n,m}'$. 
Then $\dist(x,y')+1\geq\dist(c_i,x)+\dist(y,c_{i+1}) > \dist(x,y')-1$.
\end{lemma}

\begin{proof}
By symmetry and the triangle inequality, we get
\begin{align*} \label{equa}
\dist(c_i,x)+\dist(y,c_{i+1})& > \dist(c_i,x)+\dist(y,c_{i})-1=\dist(c_i,x)+\dist(y',c_{i})-1\\
 &\geq \dist(x,y')-1.
\end{align*}
On the other hand, consider the intersection $P$ of $xy'$ with $c$ and let $c_j$ be the next vertex left of $P$. 
Since the trip belongs to an optimal tour we have $\dist(c_i,x)+\dist(y,c_{i+1})\leq \dist(c_j,x)+\dist(y,c_{j+1})$, 
otherwise we can replace $(c_i,x),(y,c_{i+1})$ by $(c_j,x),(y,c_{j+1})$ and $(c_i,c_{i+1})$ to get a pseudo-tour that is at most 1 longer than an optimum tour contradicting Lemma~\ref{lemma:pseudo-tour}
or remove in addition $(c_{j},c_{j+1})$ to get a shorter tour. Hence
\begin{align*}
\dist(c_i,x)+\dist(y,c_{i+1})&\leq \dist(c_j,x)+\dist(y,c_{j+1}) \\
&\leq \dist(x,P)+\dist(c_j,P)+\dist(P,y)+\dist(P,c_{j+1})\\
&=\dist(x,y')+1.
\end{align*} 
\end{proof}

\begin{figure}[ht]
\centering
\begin{tikzpicture}[scale=1]
\def\scale{9.5}
\def\xscale{0.45}
\def\yscale{0.65}

\def\sqrt3{1.732050808}

\def\dotMarkRightAngle[size=#1](#2,#3,#4){%
 \draw ($(#3)!#1!(#2)$) -- 
       ($($(#3)!#1!(#2)$)!#1!270:(#2)$) --
       ($(#3)!#1!(#4)$);
}

\tikzstyle{vertex}=[black,fill,circle,draw,minimum size=4,inner sep=1pt]
\tikzstyle{tedge}=[line width=1.3]
\tikzstyle{sedge}=[line width=1.3, bend above]

\coordinate[vertex, label=left: \small{$A$}] (A) at (0,0);
\coordinate[vertex, label=right:\small{$B$}] (B) at (\scale,0);
\coordinate[vertex, label=above:\small{$M$}] (M) at (\scale / 2, \scale * \sqrt3 / 6) {};
\coordinate[vertex, label=above left: \small{$x$}] (x) at (\xscale * \scale / 2, \xscale * \scale * \sqrt3 / 6) {};
\coordinate[vertex, label=below left: \small{$x'$}] (xprime) at (\xscale * \scale / 2, - \xscale * \scale * \sqrt3 / 6) {};
\coordinate[vertex, label=above left: \small{$y$}] (y) at (\yscale * \scale / 2, \yscale * \scale * \sqrt3 / 6) {};
\coordinate[vertex, label=below left: \small{$y'$}] (yprime) at (\yscale * \scale / 2, - \yscale * \scale * \sqrt3 / 6) {};
\coordinate[vertex, label=below:\small{$M'$}] (Mprime) at (\scale / 2, - \scale * \sqrt3 / 6) {};

\coordinate[vertex, label=below left:\small{$P$}] (P) at ($(A)!(x)!(Mprime)$);
\dotMarkRightAngle[size=2.5mm](A,P,x);

\draw[-,  line width=0.75] (A) to node[below right] {\small{$e$}} (M);
\draw[-,  line width=0.75] (M) to node[below left] {\small{$f$}} (B);
\draw[-,  line width=0.75] (A) to (B);
\draw[-,  line width=0.75] (A) to (Mprime);
\draw[-,  line width= 0.5, gray] (P) to (x);
\draw[-, dashed, lightgray, line width= 0.5] (x) to (xprime);
\draw[-, dashed, lightgray, line width= 0.5] (y) to (yprime);
\draw[-, dashed, lightgray, line width= 0.5] (M) to (Mprime);
\draw[-, red, line width= 1] (x) to (yprime);

\end{tikzpicture} 
 \caption{The situation in the proof of Lemma~\ref{righttriangle}.}
  \label{trickfig}
\end{figure}
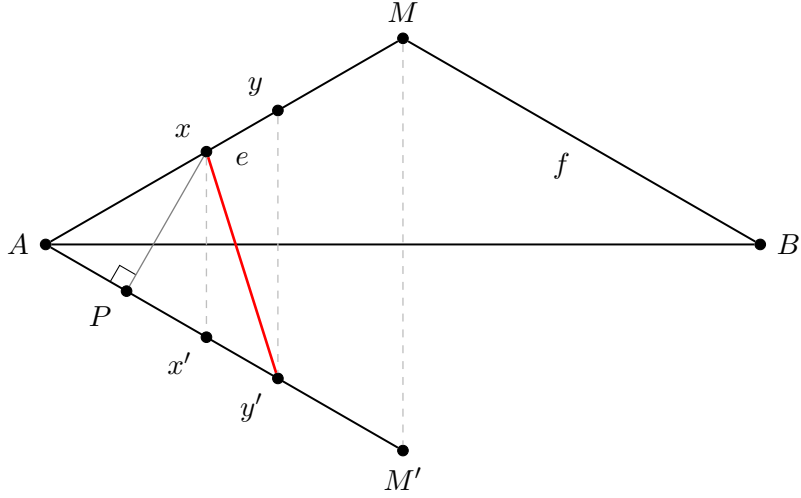

\begin{lemma} \label{righttriangle}
For $x$ and $y$ lying on the same internal segment we have:
\begin{align*}
\dist(x,y') \geq \dist(x,y)+2+2\gamma.
\end{align*}
\end{lemma}

\begin{proof} Wlog we may assume that $x,y$ are on $e$ and $x$ is to the left of $y$.
Let $P$ be the perpendicular foot of $x$ on $AM'$ (Figure~\ref{trickfig}). 
For reasons of symmetry, we get $\dist(A,x)=\dist(A,x')$ and $\angle x'Ax=2\cdot \angle BAM=\angle BAC=60^\circ$. Therefore, the triangle $Ax'x$ is equilateral and we have $\dist(P,x')=\frac{1}{2}\dist(A,x)$. Together with the triangle inequality and as the triangle $Py'x$ has the hypotenuse $xy'$ 
and~(\ref{mindistproperty}) we get:
\begin{align*}
\dist(x,y')> \dist(P,y')
&=\dist(P,x')+\dist(x',y')=\dist(x,y)+\frac{1}{2}\dist(A,x) \\
&\geq \dist(x,y)+\frac{1}{2}(4+4\gamma)=\dist(x,y)+2+2\gamma.
\end{align*}
\end{proof}

\begin{figure}[ht]
\centering
\begin{tikzpicture}[scale=1]
\def\scale{9.5}
\def\ciscale{0.35}
\def\ciiscale{0.45}

\def\yscale{0.75}
\def\xscale{0.45}
\def\pscale{0.65}
\def\qscale{0.50}
\def\sscale{0.43} 
\def\tscale{0.25} 

\def\sqrt3{1.732050808}

\def\dotMarkRightAngle[size=#1](#2,#3,#4){%
 \draw ($(#3)!#1!(#2)$) -- 
       ($($(#3)!#1!(#2)$)!#1!90:(#2)$) --
       ($(#3)!#1!(#4)$);
}

\tikzstyle{vertex}=[black,fill,circle,draw,minimum size=4,inner sep=1pt]
\tikzstyle{tedge}=[line width=1.3]
\tikzstyle{sedge}=[line width=1.3, bend above]

\coordinate[vertex, label=left: \small{$A$}] (A) at (0,0);
\coordinate[vertex, label=right:\small{$B$}] (B) at (\scale,0);
\coordinate(C) at (\scale / 2 , \scale * 0.7 * \sqrt3 / 2);  

\coordinate[vertex] (s) at (\scale / 2, \sscale * \scale) {};
\coordinate[vertex, label=above left:\small{$M$}] (M) at (\scale / 2, \scale * \sqrt3 / 6) {};
\coordinate[vertex, label=below:\small{$c_i$}] (ci) at (\ciscale * \scale, 0) {};
\coordinate[vertex, label=below:\small{$c_{i+1}$}] (cii) at (\ciiscale * \scale, 0) {};
\coordinate[vertex, label=above left: \small{$y$}] (y) at (\yscale * \scale / 2, \yscale * \scale * \sqrt3 / 6) {};
\coordinate[vertex, label=above left: \small{$x$}] (x) at (\xscale * \scale / 2, \xscale * \scale * \sqrt3 / 6) {};
\coordinate[vertex] (t) at (\tscale * \scale / 2, \tscale * \scale * \sqrt3 / 6) {};
\coordinate[vertex, label=above right: \small{$p$}] (p) at (\scale - \pscale * \scale / 2, \pscale * \scale * \sqrt3 / 6) {};
\coordinate[vertex, label=above right: \small{$q$}] (q) at (\scale - \qscale * \scale / 2, \qscale * \scale * \sqrt3 / 6) {};

\draw[-,  line width=0.75] (A) to (t);
\draw[-,  line width=0.75] (q) to (B);
\draw[-,  line width=0.75] (M) to (C);
\draw[-,  line width=0.75] (M) to (y);
\draw[-,  line width=0.75] (M) to (p);
\draw[-,  line width=0.75] (A) to (ci);
\draw[-,  line width=0.75] (B) to (cii);

\draw[-, red, line width= 1] (y) to (p) to (s) to (t) to (x);
\draw[-, myblue, line width= 1] (ci) to (x);
\draw[-, myblue, line width= 1] (cii) to (y);
\draw[-, mygreen, line width= 1] (p) to (q);
\draw[-, mygreen, line width= 1] (p) to[bend right=40] (q);
\draw[-, mygreen, line width= 1] (x) to (y);
\draw[-, mygreen, line width= 1] (ci) to (cii);

\end{tikzpicture} 
  		\caption{Illustration of the proof of Theorem~\ref{long}. The trip consists of the red and blue edges. The blue edges are replaced by the green edges to get an upper bound on the length of the trip.}
  		\label{figlongc1}
\end{figure}
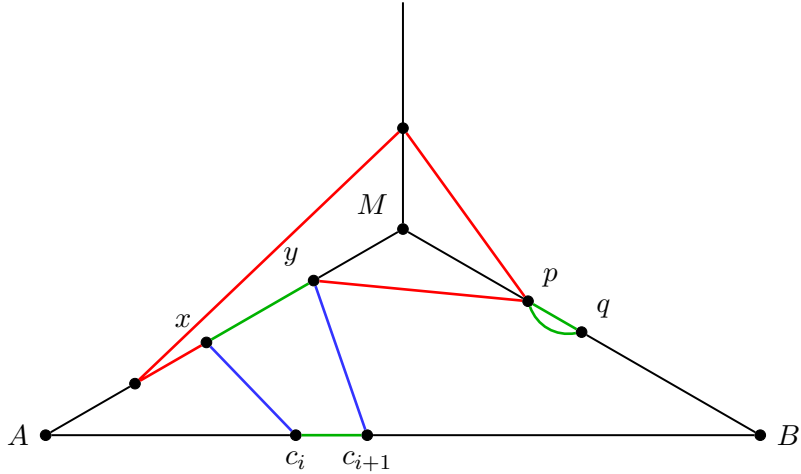

\begin{theorem} \label{long}
An optimum tour for $T'_{n,m}$ does not contain a trip where the two connection vertices lie on the same internal segment.
\end{theorem}

\begin{proof}
Assume such a trip $t$ exists. Wlog let $t$ have $c$ as main side, starting at $c_i$ and ending at $c_{i+1}$. Let $x$ resp.\ $y$ be the first resp.\ last connection vertex of the trip. By Lemma~\ref{lemma:nogi} the vertices $x$ and $y$ cannot belong to $\{g_1, \ldots, g_{m-1}\}$. We may assume that $x$ and $y$ lie wlog on $e$. 

Suppose there are internal vertices $p$ and $q$ with $\dist(p,q)=\gamma$ and $p$ is visited by the trip $t$, but $q$ is not. In this case, we replace $(c_i,x)$ and $(y,c_{i+1})$ by $(p,q)$, $(q,p)$, $(y,x)$, and $(c_i,c_{i+1})$ (Figure~\ref{figlongc1}). 
Since $q$ is not visited by the trip $t$, we get an upper bound for the length of an optimal tour. By Lemma~\ref{firstlast} and Lemma~\ref{righttriangle} we have 
\begin{align*}
\dist(c_i,x)+\dist(y,c_{i+1}) & >  \dist(x,y)+ 1 + 2 \gamma \\
                              & =  \dist(x,y) + \dist(c_i,c_{i+1}) + \dist(p,q) + \dist(q,p).
\end{align*}
Thus, the modification makes the tour shorter, contradiction. Hence, there is no such vertex $p$ visited by the trip with a neighbor $q$ at distance $\gamma$
not visited by the trip. It follows that the trip visits all internal vertices and it is the only trip.

Since the trip visits all internal vertices, there is an edge $(u, v)$ leaving the internal segment $e$ the first time where $u$ lies to the left of $y$. Furthermore $v$ cannot lie on the internal segment $f$, since otherwise $(u,v)$ would intersect the edge $(y,c_{i+1})$
Therefore, $v$ lies on the internal segment $g$. Consider $b_{n-i}$, the reflection of $c_{i}$ with respect to $e$. Because there is only one trip, $b_{n-i}$ is not a starting vertex of a trip and the edge $(b_{n-i},b_{n-(i+1)})$ is part of the tour. 
We replace the edges $(c_{i},x)$, $(y,c_{i+1})$ and $(b_{n-i},b_{n-(i+1)})$ by $(b_{n-i},x)$, $(y,b_{n-(i+1)})$ and $(c_i,c_{i+1})$. 
We have by symmetry:
\begin{align*}
\dist(c_{i},x)+\dist(y,c_{i+1})+\dist(b_{n-i},b_{n-(i+1)})\\
= \dist(b_{n-i},x)+\dist(y,b_{n-(i+1)})+\dist(c_i,c_{i+1}).
\end{align*}
Therefore, the new tour is optimum as well. But $(y,b_{n-(i+1)})$ intersects $(u, v)$, contradicting the optimality.
\end{proof}

\begin{corollary} \label{lageconn}
Let $t$ be a trip  with main side $c$ in an optimum tour for $T'_{n,m}$.
Then the first connection vertex of $t$ lies on $e$ and the second connection vertex lies on $f$.
\end{corollary}

\begin{proof}
By Theorem~\ref{long}, the two connection vertices lie on different internal segments. 
By Lemma~\ref{lemma:nogi} the first connection vertex has to lie on $e$ and the second connection vertex on $f$.
\end{proof}

The next step is to show that the optimal tour consists of only one trip of a certain shape. Let $i_0$ be the smallest integer such that $e_{i_0}$ is a vertex in $T'_{n,m}$.

\begin{lemma} \label{ei0pos}
We have 
\begin{align*}
\frac{\dist(A,e_{i_0})}{\dist(A,M)} < \frac{1}{2}.
\end{align*}
\end{lemma}

\begin{proof}
By definition of the instances$T_{n,m}'$ we 
have $\dist(A,e_{i_0})\leq \max\{10,4+4\gamma\}+\gamma$ and $\dist(A,M)=\frac{n}{\sqrt{3}} = m\cdot \gamma$.
Moreover $\frac{10+\gamma}{m\cdot \gamma} = \frac{10\cdot \sqrt{3}}{n} + \frac{1}{m} < \frac{1}{2}$ and
$\frac{4+5\cdot \gamma}{m\cdot \gamma} 
 = \frac{4\cdot\sqrt{3}}{n} + \frac{5}{m} < \frac{1}{2}$ as by assumption~(\ref{eq:nm}) we have $n\ge 40$ and $m\ge 22
$.
\end{proof}

Let $\mathcal{T}$ be the set of trips that occur in any optimum tour. For $t\in \mathcal{T}$ let $g(t)$ be the sum of the lengths of the first and last edge of $t$, and $\delta := \min_{t\in \mathcal{T}}(g(t))$. 
Let $M_c$ be the center of $c$.

\begin{figure}[ht]
\centering
\begin{tikzpicture}[scale=1]
\def\scale{9.5}
\def\fmscale{0.9}
\def\yscale{0.65}
\def\qscale{1.4}   
\def\eiscale{0.3}  

\def\sqrt3{1.732050808}

\def\dotMarkRightAngle[size=#1](#2,#3,#4){%
 \draw ($(#3)!#1!(#2)$) -- 
       ($($(#3)!#1!(#2)$)!#1!90:(#2)$) --
       ($(#3)!#1!(#4)$);
}

\tikzstyle{vertex}=[black,fill,circle,draw,minimum size=4,inner sep=1pt]
\tikzstyle{tedge}=[line width=1.3]
\tikzstyle{sedge}=[line width=1.3, bend above]

\coordinate[vertex, label=left: \small{$A$}] (A) at (0,0);
\coordinate[vertex, label=right:\small{$B$}] (B) at (\scale,0);

\coordinate[vertex, label=above:\small{$M$}] (M) at (\scale / 2, \scale * \sqrt3 / 6) {};
\coordinate[vertex, label=below:\small{$M'$}] (Mprime) at (\scale / 2, - \scale * \sqrt3 / 6) {};
\coordinate[vertex, label=above right: \small{$y$}] (y) at (\scale - \yscale * \scale / 2, \yscale * \scale * \sqrt3 / 6) {};
\coordinate[vertex, label=below right: \small{$y'$}] (yprime) at (\scale - \yscale * \scale / 2, - \yscale * \scale * \sqrt3 / 6) {};
\coordinate (q) at (\scale - \qscale * \scale / 2, - \qscale * \scale * \sqrt3 / 6) {};
\coordinate[vertex, label=above left: \small{$e_{i_0}$}] (ei) at (\eiscale * \scale / 2, \eiscale * \scale * \sqrt3 / 6) {};
\coordinate[vertex, label=above right: \small{$f_{m-1}$}] (fm) at (\scale - \fmscale * \scale / 2, \fmscale * \scale * \sqrt3 / 6) {};

\coordinate[vertex, label=above left:\small{$P$}] (P) at ($(A)!(Mprime)!(M)$);
\coordinate (eifoot) at ($(Mprime)!(ei)!(q)$);
\dotMarkRightAngle[size=2.5mm](ei,eifoot,q);
\dotMarkRightAngle[size=2.5mm](Mprime,P,M);

\draw[-,  line width=0.75] (A) to (M);
\draw[-,  line width=0.75] (M) to (B);
\draw[-,  line width=0.75] (A) to (B);
\draw[-,  line width=0.75] (A) to (Mprime);
\draw[-,  line width=0.75] (B) to (Mprime);
\draw[-,  line width=0.3]  (q) to (Mprime);

\draw[-,  line width= 0.5, gray] (P) to (Mprime);
\draw[-,  line width= 0.5, gray] (ei) to (eifoot);
\draw[-, dashed, lightgray, line width= 0.5] (y) to (yprime);
\draw[-, dashed, lightgray, line width= 0.5] (M) to (Mprime);
\draw[-, dashed, red, line width= 1] (ei) to (fm);

\end{tikzpicture} 
 \caption{The situation in the proof of Lemma~\ref{new}.}
\label{fig:new}
\end{figure}

\begin{lemma} \label{new}
Let $t$ be a trip in an optimum tour of $T'_{n,m}$ such that $t$ contains at least one of $\{e_{i_0}, f_{i_0}, g_{i_0}\}$ 
as a connection vertex. Then the total length of the first and last edge in $t$ is at least $\dist(e_{i_0},f_{m-1})+1$.
\end{lemma}

\begin{proof}
Let wlog $(c_i,e_{i_0})$ and $(y,c_{i+1})$ be the first and last edge of such a trip and $P$ be the perpendicular foot of $M'$ to $AM$ 
(Figure~\ref{fig:new}). By Corollary~\ref{lageconn} we know that $y$ lies on $f$. 
Since the triangle $AM'M$ is equilateral, $P$ is the center of $AM$.
Therefore by Lemma~\ref{ei0pos} $e_{i_0}$ lies left of $P$. Hence the perpendicular foot of $e_{i_0}$ to $M'B$ lies left of $M'$ 
and we have $\dist(e_{i_0},y')\geq \dist(e_{i_0},M')$. 
By Lemma~\ref{firstlast} and Lemma~\ref{righttriangle} we get 
\begin{align*}
\dist(c_i,e_{i_0})+\dist(y,c_{i+1}) \geq & \dist(e_{i_0},y')-1 \geq \dist(e_{i_0},M')-1 \\
\geq &\dist(e_{i_0},M)+1+2\gamma> \dist(e_{i_0},f_{m-1})+1.
\end{align*}
\end{proof}

\begin{lemma} \label{firstlastedgeoftrip}
We have
$\delta\geq \dist(P,M')-1$ where $P$ is the perpendicular foot of $M'$ to $AM$.
\end{lemma}

\begin{proof}
Let $x$ and $y$ be the first and last connection vertex of a trip contained in an optimum tour of $T_{n,m}'$. 
By Lemma~\ref{firstlast} the total length of the first and last edge of this trip is at least $\dist(x,y')-1$. 
By Corollary~\ref{lageconn} $\dist(x,y')\geq \dist(P,M')$, since $\dist(P,M')$ is the distance between the parallel segments $AM$ and $M'B$.
\end{proof}

\begin{lemma} \label{new2}
We have
$\dist(g_{i_0},f_{m-1})+\delta > \min_i(\dist(c_i,e_{i_0})+\dist(f_{i_0},c_{i+1}))+\gamma$
\end{lemma}

\begin{proof}
By Lemma~\ref{firstlastedgeoftrip} it suffices to prove that 
$\dist(g_{i_0},f_{m-1})+\dist(P,M')-1> \min_i(\dist(c_i,e_{i_0})+\dist(f_{i_0},c_{i+1}))+\gamma$.
As by the triangle inequality $\min_i(\dist(c_i,e_{i_0})+\dist(f_{i_0},c_{i+1})) = \min_i(\dist(c_i,e_{i_0})+\dist(f_{i_0},c_{i+1})+\dist(c_{i},c_{i+1}))-1 < \dist(e_{i_0}, f_{i_0}') + 1$
the latter inequality follows from the inequality 
$\dist(g_{i_0},M)+\dist(P,M') -2 -\gamma > \dist(e_{i_0},f'_{i_0})$ which we are now going to prove.

Consider the point on $S$ on $AM$ with $\dist(M,S)=\dist(P,M')-2-\gamma$ and that lies on the opposite side of $M$ than $A$. It is enough to show that $\frac{\dist(e_{i_0},f'_{i_0})}{\dist(e_{i_{0}},S)}< 1$. By symmetry, we have $2\cdot \dist(e_{i_0},M_c)=\dist(e_{i_0},f'_{i_0})$. 
By the sine law in the triangle $e_{i_0}M_cS$, we have $\frac{\dist(e_{i_0},M_c)}{\dist(e_{i_0},S)}=\frac{\sin(\angle e_{i_0}SM_c)}{\sin(\angle SM_ce_{i_0})}$. If we move $e_{i_0}$ towards $A$, $\angle e_{i_0}SM_c$ is fixed and $\angle SM_ce_{i_0}$ is monotonically increasing. Since the sine is concave in $[0,\pi]$, it is enough to verify the statement for the cases where $e_{i_0}$ is as far and as near as possible to $A$. 
By Lemma~\ref{ei0pos} 
$e_{i_0}$ lies between $A$ and $P$. For $e_{i_0}=A$, we have using~(\ref{eq:nm}): 
\begin{align*}
\dist(e_{i_0},S)&=\dist(e_{i_0},M)+\dist(M,S)=\frac{\sqrt{3}}{3}n+\frac{1}{2}n-2-\gamma\\
&\geq \frac{\sqrt{3}}{3}n+\frac{1}{2}n -2-\frac{\sqrt{3}}{66}n > n+\frac{1}{20}n-2\geq n > \dist(e_{i_0},f'_{i_0}).
\end{align*}
For $e_{i_0}=P$, we have 
\begin{align*}
\dist(e_{i_0},S) = \dist(e_{i_0},M)+\dist(M,S)& \geq\frac{1}{2\sqrt{3}}n + \frac{1}{2}n-2-\gamma\\
&\geq \frac{1}{\sqrt{3}}n+\frac{\sqrt{3}-1}{2\sqrt{3}}n-2-\frac{\sqrt{3}}{66}n \\
& > \frac{1}{\sqrt{3}}n = \dist(A,M) = \dist(e_{i_0},f'_{i_0}).
\end{align*}
\end{proof}

\begin{lemma} \label{lemma:optstart}
$c_{\lfloor \frac{n-1}{2} \rfloor}$ resp.\ $c_{\lfloor \frac{n-1}{2}\rfloor+1}$ and $c_{\lceil \frac{n-1}{2} \rceil}$ resp.\ $c_{\lceil \frac{n-1}{2}\rceil+1}$ are the optimal starting resp.\ ending vertices of a trip with main side $c$ and connection vertices $e_{i_0}$ and $f_{i_0}$.
\end{lemma}

\begin{proof}
We have to show $\dist(c_i,e_{i_0})+\dist(f_{i_0},c_{i+1})$ is minimal for 
$i\in\{\lfloor \frac{n-1}{2} \rfloor, \lceil \frac{n-1}{2} \rceil\}$. Let $\overline{f}_{i_0}$ be the point obtained by shifting $f_{i_0}$ to the left by 1. Let $\overline{M}$ be the point of intersection of the perpendicular bisector of $e_{i_0}\overline{f}_{i_0}$ and $c$. We have $\dist(c_i,e_{i_0})+\dist(f_{i_0},c_{i+1})=\dist(c_i,e_{i_0})+\dist(\overline{f}_{i_0},c_{i})$.
The trace of points $X$ satisfying $\dist(X,e_{i_0})+\dist(\overline{f}_{i_0},X)=k$ for a fixed $k$ is an ellipsoid with focal points $e_{i_0}$ and $\overline{f}_{i_0}$. The size of the ellipsoid is strictly monotonically decreasing with $k$. Consider the ellipsoid through $c_i$ with focal points $e_{i_0}$ and $\overline{f}_{i_0}$. The size of the ellipsoid is strictly monotonically decreasing with the distance of $c_i$ to $\overline{M}$. Thus, the sum $\dist(c_i,e_{i_0})+\dist(\overline{f}_{i_0},c_{i})$ is strictly monotonically decreasing with the distance of $c_i$ to $\overline{M}$. For odd $n$, $\overline{M}$ is $c_{\frac{n-1}{2}}$, for even $n$, it is the midpoint of $c_{\lfloor \frac{n-1}{2} \rfloor}$ and $c_{\lceil \frac{n-1}{2} \rceil}$. This proves the statement. 
\end{proof}

\begin{lemma} \label{key}
Let $x$ resp.\ $y$ be vertices on $e$ resp.\ $f$, $c_i$ be any base vertex. 
If the edge $(c_i,x)$ is part of an optimum tour in $T_{n,m}'$, then there exists no $z$ on $e$ and to the left of $x$ such that $(z,g_k)$ is part of the tour for any vertex $g_k$. 
Similarly if $(y,c_{i})$ is part of the tour, then there exists no $z$ on $f$ and to the right of $y$ such that $(g_k,z)$ is part of the tour for any vertex $g_k$.
\end{lemma}

\begin{proof}
Assume there is a leaving edge $(z,g_k)$ with $z$ to the left of $x$. Then, $(c_i, z)$ is not in the tour, since this would result in
a pseudo-tour and $(x,g_k)$ is not in the tour since otherwise the edge $(z,g_k)$ would be oriented $(g_k,z)$.
Hence, we can replace the edges $(c_i, x)$ and $(z,g_k)$ by $(c_i, z)$ and $(x, g_k)$. 
If the old tour was $(c_i,x)X(z,g_k)Y$ where $X, Y$ are subpaths, we get the new tour $(c_i,z)\overline{X}(x,g_k)Y$. Here $\overline{X}$
denotes the reversed subpath of $X$. Since the segments $c_ix$ and $f_kz$ intersect, $\dist(c_i,z)+\dist(f_k,x) < \dist(c_i, x)+\dist(f_k, z)$
follows from the (strict) triangle inequality applied to the triangles $zc_is$ and $sf_kx$, where $s$ is the point of intersection of $c_ix$
and $f_kz$. By symmetry we have $\dist(f_k,x)=\dist(g_k,x)$ and $\dist(f_k, z)=\dist(g_k, z)$ and thus 
$\dist(c_i,z)+\dist(g_k,x) < \dist(c_i, x)+\dist(g_k, z)$.  
Hence, the new tour is shorter, which is a contradiction to the optimality. Analogously we can prove the second statement.
\end{proof}

\begin{lemma}
For all internal vertices, the distance to the nearest vertex is $\gamma$, for all other vertices it is 1.
\end{lemma}

\begin{proof}
By definition the smallest distance from an internal vertex to another internal vertex is $\gamma$, and the smallest distance from a
noninternal vertex to another noninternal vertex is 1. On the other hand, the smallest distance between an internal vertex and a noninternal
vertex is at least $\sin(30^\circ)\cdot \dist(A,e_{i_0})\ge \frac{1}{2}\max\{10,4+4\gamma\}$ which is larger than 1 and $\gamma$.
\end{proof}

Consider the tour $T^*$ that only contains one trip which visits $e_{i_0}$ first and $f_{i_0}$ last and visits adjacent internal vertices
except for the edge $(g_{i_0},f_{m-1})$. 
Moreover, the first and last vertex of the trip in $T^*$ are the vertices 
$c_{\lfloor \frac{n-1}{2} \rfloor}$ and $c_{\lfloor \frac{n-1}{2}\rfloor+1}$ (Figure~\ref{fig:OptimumTour}).
We show that up to rotations and reflections the tour $T^*$ is the only
optimum tour in $T_{n,m}'$. 

\begin{figure}[ht]
\centering

\begin{tikzpicture}[scale=0.3]
\def\n{24}
\def\m{18}
\def\sqrt3{1.73205}
\def\mycircle{circle(1.25mm)}

\def\a#1{(\n     - #1 / 2, #1 * \sqrt3 / 2)}
\def\c#1{(#1, 0)}
\def\e#1{(#1 * \n / 2 / \m, #1 * \n / 2 / \sqrt3 / \m)}
\def\f#1{(\n - #1 * \n / 2 / \m, #1 * \n / 2 / \sqrt3 / \m)}
\def\g#1{(\n / 2, \n * \sqrt3 / 2 - #1 * \n / \sqrt3 / \m)}

\draw[red]  \c0 -- \c{12};
\draw[red]  \c{12} -- \e{9} -- \e{18} -- \g{9} -- \f{17} -- \f{9} -- \c{13};
\draw[red]  \c{13} -- \a{0} -- \a{24} -- \c{0};

\pgfmathparse{(\n - 1)}
\foreach \i in {1,...,\pgfmathresult}
{
   \fill (\n     - \i / 2, \i * \sqrt3 / 2) \mycircle ;
   \pgfmathparse{(\n - \i) * \sqrt3 / 2}
   \fill (\n / 2 - \i / 2,  \pgfmathresult) \mycircle;
   \fill (             \i,               0) \mycircle;
}

\pgfmathparse{\m - 1}
\foreach \j in {9,...,\pgfmathresult}
{
   \fill(     \j * \n / 2 / \m, \j * \n / 2 / \sqrt3 / \m) \mycircle ;
   \fill(\n - \j * \n / 2 / \m, \j * \n / 2 / \sqrt3 / \m) \mycircle ;
   \fill(               \n / 2, \n * \sqrt3 / 2 - \j * \n / \sqrt3 / \m) \mycircle ;
}

\fill(\n / 2 ,  \n / 2 / \sqrt3) \mycircle ;
\fill(\n, 0) \mycircle ;
\fill(\n / 2,  \n * \sqrt3 / 2) \mycircle ;
\fill(0, 0) \mycircle ;

\end{tikzpicture}

\caption{An optimum tour for the instance  $T'_{24,18}$.}
\label{fig:OptimumTour}
\end{figure}

Let $K := 3n -1 + (3(m-i_0)-1)\gamma$. This is the maximum total length of all pairwise point distances that are $1$ or $\gamma$
that can be contained in a tour in $T_{n,m}'$.
Note that the tour $T^*$ has length
 $\dist(c_{\lfloor \frac{n-1}{2} \rfloor},e_{i_0})+\dist(f_{i_0},c_{\lfloor \frac{n-1}{2}\rfloor+1})+\dist(g_{i_0},f_{m-1})+ K$.

Let the \emph{long edges} of a TSP tour be the edges that are incident to one of $\{e_{i_0},f_{i_0},g_{i_0}\}$
and have length larger than $\gamma$.
A long edge is called \emph{internal} if it connects two internal vertices, it is called \emph{externalizing} if it is the 
starting or ending edge of a trip.

\begin{theorem}
Up to symmetry the tour $T^*$ is the unique optimum tour for $T'_{n,m}$.
\label{thm:uniqueOpt}
\end{theorem}

\begin{proof} Let $T$ be an optimum tour. First we prove that there exist at least three long edges in $T$. 
Consider the neighbors of $e_{i_0}$ in $T$. As the tour $T$ is a simple polygon by Lemma~\ref{lemma:nocrossing} 
one of the two edges incident to  $e_{i_0}$ is either an internal or externalizing long edge. 
With the same argument for $f_{i_0}$ and $g_{i_0}$ we get at least two long edges incident to $\{e_{i_0},f_{i_0},g_{i_0}\}$. 
Assume that we get exactly two long edges. 
Then, we have at least one long edge connecting two of $\{e_{i_0},f_{i_0},g_{i_0}\}$, wlog the edge $\{e_{i_0},g_{i_0}\}$. 
In this case $e_{i_0}$ or $g_{i_0}$ cannot be connection vertices, since otherwise we would get three long edges. 
Hence, the trip containing the edge $\{e_{i_0},g_{i_0}\}$ cannot have $b$ as main side, since this would intersect the edge. By symmetry, assume that $(c_i,x)$ and $(y,c_{i+1})$ are the first and last edge of the trip, respectively. By Corollary~\ref{lageconn} $x$ lies on $e$. If the orientation of the edge $\{e_{i_0},g_{i_0}\}$ is $(e_{i_0},g_{i_0})$ then $(c_i,x)$ and $(e_{i_0},g_{i_0})$ contradict Lemma~\ref{key}, since $e_{i_0}$ and $g_{i_0}$ are not connection vertices. 
If the orientation is $(g_{i_0},e_{i_0})$, the part of the tour connecting $c_{i+1}$ and $e_{i_0}$ has to intersect that part connecting $c_{i}$ and $g_{i_0}$. This contradicts the optimality of the tour. 
Hence, we get at least three long edges. 

Given a set of long edges, we can get a lower bound of the tour by bounding the length of the remaining edges by the distance of the vertices 
to the closest neighbors. If two of them are externalizing long edges sharing a trip and the third is internal, then the tour has at least 
length 
$K + \dist(c_{\lfloor \frac{n-1}{2} \rfloor},e_{i_0}) + \dist(f_{i_0},c_{\lfloor \frac{n-1}{2}\rfloor+1}) + \dist(g_{i_0},f_{m-1})$ 
which is the length of $T^*$.
If the third edge is an externalizing edge, the tour is longer than $K - 1 + \dist(c_i,e_{i_0}) + \dist(f_{i_0},c_{i+1})
+ \dist(e_{i_0},f_{m-1}) + 1$
by Lemma~\ref{new}. This is more than the length of $T^*$. It remains the case where no two of these three edges belong to the same trip. 
If there is one externalizing and two internal edges the tour has length at least $K - \gamma + 2\dist(g_{i_0},f_{m-1})+\delta$.
By Lemma~\ref{new2} this is larger than the length of $T^*$. 
If all three long edges are internal edges, we get an even higher lower bound for the length of $T$ of value 
$K - 2\gamma + 3 \dist(g_{i_0},f_{m-1}) + \delta$ since every tour has at least one trip. 
Tours where at least two of the three long edges are externalizing edges have at least length 
$\min(K - 1 - \gamma + \dist(g_{i_0},f_{m-1}) + 2\delta, K - 2 - \gamma + 3\delta)$ which is larger than the 
length of $T^*$ by Lemma~\ref{new} and Lemma~\ref{new2}. 
Together with Lemma~\ref{lemma:optstart} the tour $T^*$ is unique up to symmetry.
\end{proof}

Optimum tours in $T_{n,m}$ are not unique and contain several trips. An example of an optimum tour in $T_{n,m}$ is shown in 
Figure~\ref{fig:opt_tour_in_Tnm}.

\begin{figure}[ht] \center
\includegraphics[width=0.6\textwidth]{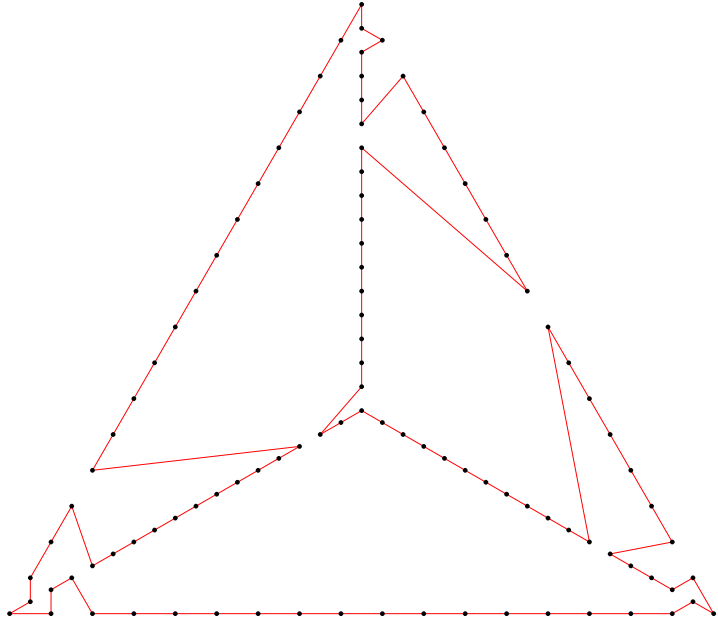}
\caption{\small An optimum solution for the instance $T_{17,17}$.}
\label{fig:opt_tour_in_Tnm}
\end{figure}

\subsection{Asymptotic Value of the Integrality Ratio}

To get a lower bound on the integrality ratio of the instances $T_{n,m}$ we first compute a lower bound on the length
of an optimum TSP tour of the instance $T'_{n,m}$ and an upper bound on an optimum solution to the subtour LP for the
instance $T'_{n,m}$.

\begin{theorem} \label{thm:OptTourBound}
For $n \le 3/2 \cdot m$ an optimum TSP tour of the instance $T'_{n,m}$ has length at least $ 4n + 4 n / \sqrt 3 - 69 $
and at most $ 4n + 4 n / \sqrt 3 - 17$.
\end{theorem}

\begin{proof}
From the proof of Theorem~\ref{thm:uniqueOpt} it follows that an optimum TSP tour for $T'_{n,m}$ has length
$$3n - 1 + 3\cdot \dist(M, e_{i_0}) - \gamma  + \dist(g_{i_0}, f_{m-1}) + \dist(c_{\lfloor \frac{n-1}{2}\rfloor}, e_{i_0}) + \dist(c_{\lfloor \frac{n-1}{2}\rfloor + 1}, f_{i_0}), $$
where $i_0$ is the smallest index such that $e_{i_0}$ is a vertex of $T'_{n,m}$. 
For $n \le 3/2 \cdot m$ we have $ \gamma = \frac{n}{\sqrt 3\cdot m} \le \frac{\sqrt 3}{2}$ and therefore $\max\{10, 4 + 4 \gamma\} = 10$.
This implies  $10 \le \dist(A,e_{i_0}) < 10 + \gamma$
and using the triangle inequality we get
$$n/\sqrt 3 - 10 + \gamma > \dist(g_{i_0}, f_{m-1}) > \dist(M, e_{i_0}) > n/\sqrt 3 - 10 - \gamma .$$
Applying the triangle inequality once more we have:
$$ \dist(A,c_{\lfloor \frac{n-1}{2}\rfloor}) + (10 + \gamma)  > \dist(A,c_{\lfloor \frac{n-1}{2}\rfloor}) + \dist(A,e_{i_0}) > 
\dist(c_{\lfloor \frac{n-1}{2}\rfloor}, e_{i_0})  
$$ and
$$ 
\dist(c_{\lfloor \frac{n-1}{2}\rfloor}, e_{i_0})  
 \ge \dist(A,c_{\lfloor \frac{n-1}{2}\rfloor}) - \dist(A,e_{i_0}) 
 \ge \dist(A,c_{\lfloor \frac{n-1}{2}\rfloor}) - (10 + \gamma). 
$$
Combining these inequalities with similar inequalities for $\dist(c_{\lfloor \frac{n-1}{2}\rfloor + 1}, f_{i_0})$ we get:
\begin{eqnarray*}
\dist(c_{\lfloor \frac{n-1}{2}\rfloor}, e_{i_0}) + \dist(c_{\lfloor \frac{n-1}{2}\rfloor + 1}, f_{i_0}) & \ge & 
\dist(A,c_{\lfloor \frac{n-1}{2}\rfloor}) + \dist(c_{\lfloor \frac{n-1}{2}\rfloor + 1}, B)  - 2 (10 + \gamma) 
\\ & = & \dist(A,B) - 21 - 2 \gamma
\\ & = & n - 21 - 2 \gamma
\end{eqnarray*}  
and
\begin{eqnarray*}
\dist(c_{\lfloor \frac{n-1}{2}\rfloor}, e_{i_0}) + \dist(c_{\lfloor \frac{n-1}{2}\rfloor + 1}, f_{i_0}) & \le & 
\dist(A,c_{\lfloor \frac{n-1}{2}\rfloor}) + \dist(c_{\lfloor \frac{n-1}{2}\rfloor + 1}, B)  + 2 (10 + \gamma) 
\\ & = & \dist(A,B) + 19 + 2\gamma
\\ & = & n + 19 + 2 \gamma .
\end{eqnarray*}  

Combining all these inequalities and using $\gamma < 1$ for  $n \le 3/2 \cdot m$ we get that 
$$ 4n + 4 n / \sqrt 3 - 69 $$
is a lower bound and
$$ 4n + 4 n / \sqrt 3 - 17 $$
is an upper bound on the length of an optimum TSP tour in $T'_{n,m}$.
\end{proof}

\begin{theorem}\label{thm:SubTourBound}
For $n \le 3/2 \cdot m$ an optimum solution to the subtour LP of the instance $T'_{n,m}$ has length at most $ 3n + 3 n / \sqrt 3$
and at least $3n + 3 n/ \sqrt 3 - 33$. 
\end{theorem}

\begin{proof}
A feasible solution to the subtour LP is shown in Figure~\ref{fig:OptSubtour}. Its total length is at most $ 3n + 3 n / \sqrt 3$.
This proves the upper bound. 

\begin{figure}[ht]
\centering

\begin{tikzpicture}[scale=0.3]
\def\n{24}
\def\m{18}
\def\sqrt3{1.73205}
\def\mycircle{circle(1mm)}

\def\a#1{(\n     - #1 / 2, #1 * \sqrt3 / 2)}
\def\b#1{(\n / 2 - #1 / 2, \n * \sqrt3 / 2 - #1 * \sqrt3 / 2)}
\def\c#1{(#1, 0)}
\def\e#1{(#1 * \n / 2 / \m, #1 * \n / 2 / \sqrt3 / \m)}
\def\f#1{(\n - #1 * \n / 2 / \m, #1 * \n / 2 / \sqrt3 / \m)}
\def\g#1{(\n / 2, \n * \sqrt3 / 2 - #1 * \n / \sqrt3 / \m)}

\draw[blue] \c1 -- \c{\n};
\draw[blue] \a1 -- \a{\n};
\draw[blue] \b1 -- \b{\n};
\draw[blue] \e{9} -- \e{17};
\draw[blue] \f{9} -- \f{17};
\draw[blue] \g{9} -- \g{18};
\draw[red]  \c{0} -- \e{9} -- \c{1} -- \c{0};
\draw[red]  \a{0} -- \f{9} -- \a{1} -- \a{0};
\draw[red]  \b{0} -- \g{9} -- \b{1} -- \b{0};
\draw[red]  \e{17} -- \g{18} -- \f{17} -- \e{17};

\pgfmathparse{(\n - 1)}
\foreach \i in {1,...,\pgfmathresult}
{
   \fill (\n     - \i / 2, \i * \sqrt3 / 2) \mycircle ;
   \pgfmathparse{(\n - \i) * \sqrt3 / 2}
   \fill (\n / 2 - \i / 2,  \pgfmathresult) \mycircle;
   \fill (             \i,               0) \mycircle;
}

\pgfmathparse{\m - 1}
\foreach \j in {9,...,\pgfmathresult}
{
   \fill(     \j * \n / 2 / \m, \j * \n / 2 / \sqrt3 / \m) \mycircle ;
   \fill(\n - \j * \n / 2 / \m, \j * \n / 2 / \sqrt3 / \m) \mycircle ;
   \fill(               \n / 2, \n * \sqrt3 / 2 - \j * \n / \sqrt3 / \m) \mycircle ;
}

\fill(\n / 2 ,  \n / 2 / \sqrt3) \mycircle ;
\fill(\n, 0) \mycircle ;
\fill(\n / 2,  \n * \sqrt3 / 2) \mycircle ;
\fill(0, 0) \mycircle ;

\end{tikzpicture}

\caption{A feasible solution for the subtour LP of the instance  $T'_{24,18}$. The red lines correspond to variables with value~$1/2$ while the blue lines 
correspond to variables with value~1.}
\label{fig:OptSubtour}
\end{figure}
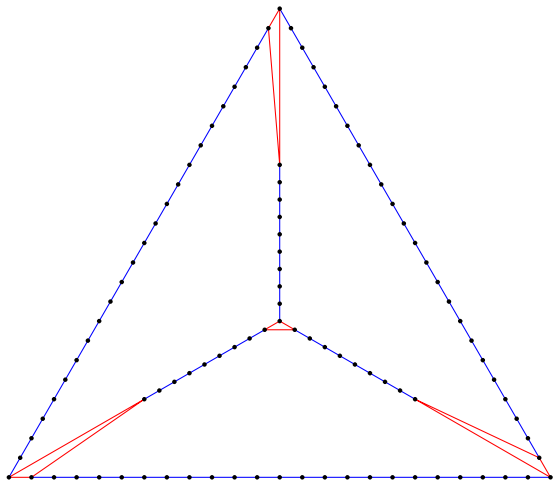
For the lower bound we observe that the nearest neighbor of a base vertex has distance~1 while a nearest neighbor of an internal vertex has distance
$\gamma$. Therefore, the total length of a feasible solution to the subtour LP must be at least 
$$ 3n + 3(n/ \sqrt 3 - 10 - \gamma) > 3n + 3 n/ \sqrt 3 - 33 .$$
\end{proof}

\begin{theorem} \label{thm:convergencerate}
For $n \le 3/2 \cdot m$ the integrality ratio of the instances $T'_{n,m}$ converges to $4/3$ for $n\to\infty$.
\end{theorem}

\begin{proof}
Using Theorem~\ref{thm:OptTourBound} and Theorem~\ref{thm:SubTourBound} we conclude that the integrality ratio of $T'_{n,m}$
is at most
$ \frac{4n + 4 n / \sqrt 3 - 17}{3n + 3 n/ \sqrt 3 - 33}$ and at least 
$ \frac{ 4n + 4 n / \sqrt 3 - 69}{ 3n + 3 n / \sqrt 3}$. Both these values converge to $4/3$ for $n\to\infty$.
\end{proof}

The length of an optimum tour in $T_{n,m}$ is clearly at least as long as an optimum tour in $T'_{n,m}$ and at most by some constant value
larger if  $n \le 3/2 \cdot m$. The bounds for a feasible solution to the subtour LP of $T'_{n,m}$ carry over to $T_{n,m}$. Therefore we get:

\begin{corollary}
For $n \le 3/2 \cdot m$ the integrality ratio of the instances $T_{n,m}$ converges to $4/3$ for $n\to\infty$.
\end{corollary}

By comparing the rate of convergence obtained in the proof of Theorem~\ref{thm:convergencerate} with Theorem~13 in~\cite{Hou2014}
one easily sees that for a given number of vertices the integrality ratio of the instances $T_{n,m}'$ converges much faster to $4/3$
than the instances constructed in~\cite{Hou2014}. Theorem~\ref{thm:convergencerate} yields the integrality ratio $\frac43 -\Theta(N^{-1})$
for an $N$-vertex instance
while the instances constructed in~\cite{Hou2014} have integrality ratio $\frac43 -\Theta(N^{-2/3})$.

\section{Experimental Results}

\label{sec:Experiments}

All runtime experiments described in this section were performed using \concorde\ version 03.12.19~\cite{ABCC2003b}.
This is the fastest known code to solve large TSP instances exactly. 
The source code of \concorde\ can be downloaded at~\cite{ABCC2003b}.
We used \texttt{gcc~4.8.5} to compile \concorde\ and used \texttt{CPLEX~12.6} as the LP solver.
All experiments were performed on a 2.20GHz Intel Xeon~E5-2699~v4 using a single core for each job. 
Up to 16 jobs were run in parallel on the machine. We have switched off the turbo boost mode of the machine
to reduce the variation in runtimes.

Repeated runs of \concorde\ on the same instance may vary a lot in their runtime if different random seeds are used for \concorde. 
On some instances we observed
a runtime difference of more than a factor of~10 between runs with different random seeds. 
In our experiments we therefore took the average runtime of $k$ runs
where $k$ was chosen to be $10$ or $100$. We used the \concorde\ option \texttt{-s \#}
to set the random seed to $i$ in the $i$-th of the $k$ runs.

\subsection{TSPLIB results}
The TSPLIB is a very well known collection of TSP instances and can be downloaded at~\cite{TSPLIBRice}.
It contains 111 instances with sizes between 14 and 85900. For all instances optimum solutions are known~\cite{ABCC2006}.
For our experiments we used all TSPLIB instances with at most 2000 vertices except the 
instance \texttt{linhp318.tsp} which contains a fixed edge.
Thus, our TSPLIB testbed contains 93~instances with sizes between 16 and 1889.
In Figure~\ref{fig:TSPLIBruntimes} we show for each of the 93 instances the minimum, average and maximum runtime
taken over 100~runs of \concorde.

It can be seen that \concorde\ solves each TSPLIB instance with less than 1000 vertices within a minute.
The slowest run took 56\,834 seconds and was on the instance \texttt{u1817}.

\begin{figure}[ht]
\center
\includegraphics[width=0.9\textwidth]{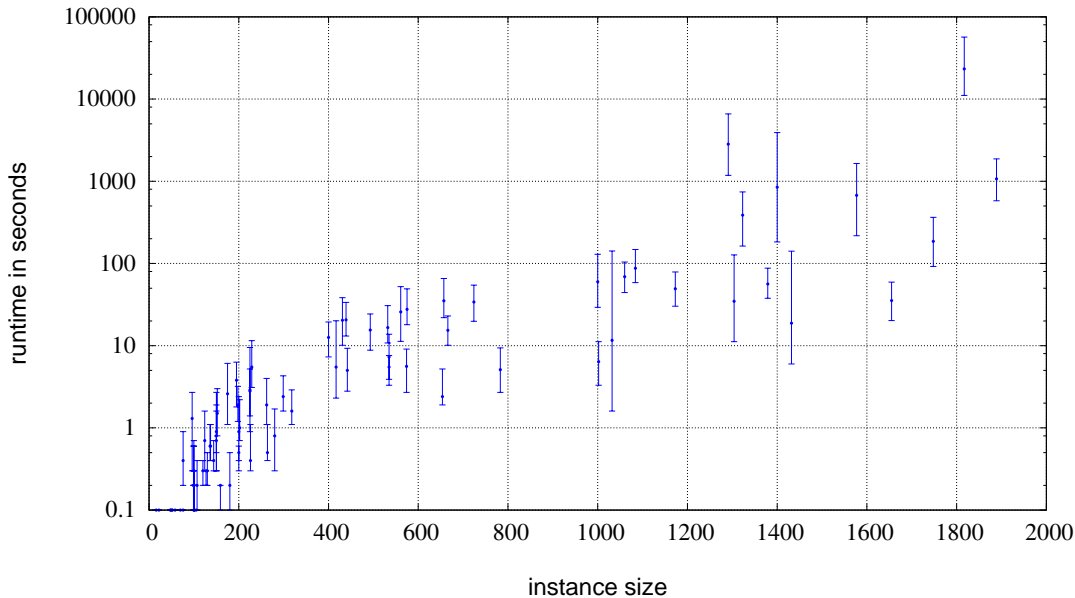}
\caption{\small Runtimes of \concorde\ on 93 TSPLIB instances with sizes between 16 and 1889.
The x-axis shows the instance size while the y-axis gives the log scaled runtime. 
For each instance a vertical line is drawn indicating the minimal and maximal runtime seen over 100
runs using different random seeds. A dot on this line marks the average runtime of these 100 runs.}
\label{fig:TSPLIBruntimes}
\end{figure}

\subsection{Results for points on three parallel lines}

In~\cite{Hou2014} a family of euclidean TSP instances is proposed that 
arises from three sets of $n$ equidistant points placed on three parallel lines.
It was shown that this family of instances has an integrality ratio converging to $4/3$
if the distances of the three lines are chosen appropriately. 
We denote by $P_{n,d}$ the instance whith $n$ points at distance~$1$ on each of the three parallel lines
which have distance $d$. We have generated these instances for all $n$ with $34 \le n \le 64$ and the   
distances $d$ ranging between $0.1$ and $10.0$ in increments of $0.1$. All point coordinates were scaled by $10000$.
Thus in total we generated 3300 instances. For each instance we measured the average runtime of 100 runs of \concorde.
From these results we chose for each $n$ the instance having the largest average runtime. 

Figure~\ref{Fig:P_runtimes} shows the minimum, average and maximum runtime for all these chosen instances. 
The red line shown in this figure is a 
least-squares fit of the average runtimes. 
 
\begin{figure}[ht]
\center
\includegraphics[width=0.9\textwidth]{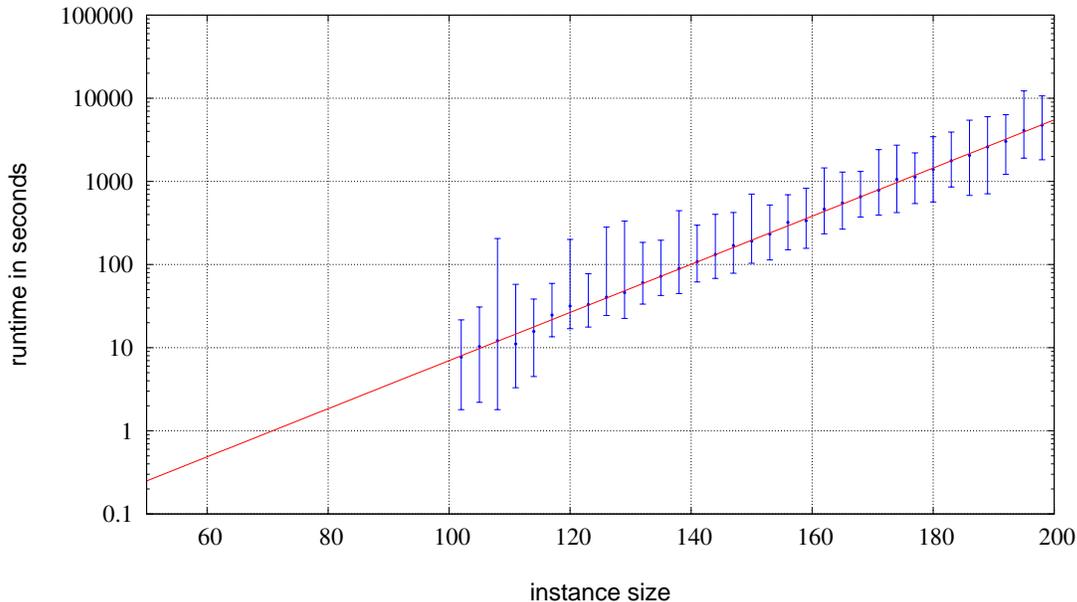}
\caption{\small Runtimes of \concorde\ on the $P_{n,d}$ instances with at least 100 and at most 200 vertices.
The distance $d$ was chosen to maximize the largest average runtime. 
The x-axis shows the instance size while the y-axis gives the log scaled runtime. 
For each instance a vertical line is drawn indicating the minimal and maximal runtime seen over 100
runs using different random seeds. A dot on this line marks the average runtime of these 100 runs. The red line is a least-squares fit 
of the average runtimes. }
\label{Fig:P_runtimes}
\end{figure}

\subsection{Results for the instances \boldmath $T_{n,m}$}

The instances $T_{n,m}$ have $N := 3(n+m) - 2$ vertices. 
We created these instances in TSPLIB format by scaling the point coordinates by 10000 and rounding to the nearest integer.
The distance between two points is defined as the rounded Euclidean distance called EUC\_2D in the TSPLIB format.
All these instances are available for download at~\cite{Tnmurl}.
We measured the runtime of \concorde\ on the instances $T_{n,m}$
for all $5 \le n \le 33$ and $5 \le m \le 33$ to get some idea for which choices of $n$ and $m$ the largest runtimes appear
if the number of vertices of the instance is fixed. Our conclusion from that experiment is that for a given number $N$ of vertices 
with $N\equiv 1 \bmod 3$ and $N \ge 50$ the following choices for $n$ and $m$ lead to high runtimes of \concorde:

\begin{equation}
n := \lfloor\frac{3 N - 40}{10}\rfloor  \mbox{~~~~~ and ~~~~~}   m := \frac{N + 2}{3} - n 
\label{eqn:nm}
\end{equation}

\begin{figure}[ht]
\center
\includegraphics[width=0.9\textwidth]{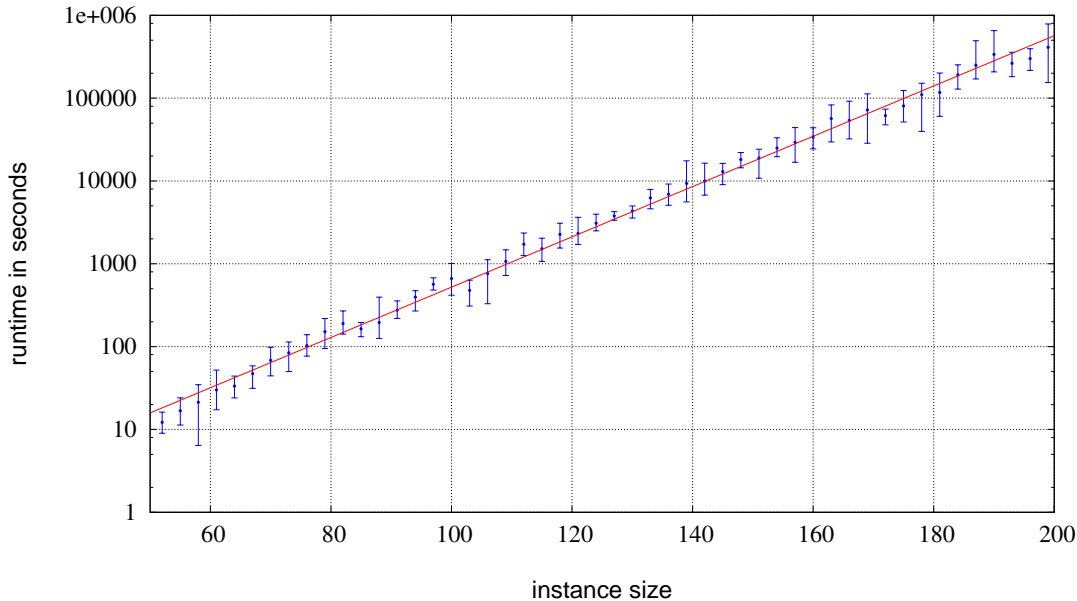}
\caption{\small Runtimes of \concorde\ on the $T_{n,m}$ instances with at least 52 and at most 200 vertices.
The values $n$ and $m$ were selected according to the equations~(\ref{eqn:nm}).
The x-axis shows the instance size while the y-axis gives the log scaled runtime. 
For each instance a vertical line is drawn indicating the minimal and maximal runtime seen over 10
runs using different random seeds. A dot on this line marks the average runtime of these 10 runs. The red line is a least-squares fit 
of the average runtimes. }
\label{fig:TnmRuntimes}
\end{figure}

Figure~\ref{fig:TnmRuntimes} shows the minimum, average and maximum runtime for all instances $T_{n,m}$ with at least 50
and at most 200 vertices and with $n$ and $m$ defined by the equations~(\ref{eqn:nm}). The red line shown in this figure is a 
least-squares fit of the average runtimes. It corresponds to the function 

\begin{equation}
\mbox{runtime in seconds} ~=~ 0.480 \cdot 1.0724^N
\label{eqn:RuntimeTnm}
\end{equation}

Thus, this function estimates the runtime in seconds needed by \concorde\ for the instances $T_{n,m}$ with 
$N = 3(n+m) - 2$ and $n$ and $m$ defined by the equations~(\ref{eqn:nm}).
From~(\ref{eqn:RuntimeTnm}) we get for example the following very rough runtime estimates:
\begin{itemize}
\item  $n=60$ and $m=12$ implies $N=214$ and runtime estimate 17 days.
\item  $n=71$ and $m=13$ implies $N=250$ and runtime estimate 216 days.
\item  $n=296$ and $m=38$ implies $N=1000$ and runtime estimate $3\cdot 10^{22}$ years.
\end{itemize}

The largest runtime that we have measured for \concorde\ on a TSPLIB instance with at most 1000 vertices was 129.2 seconds
on the instance \texttt{dsj1000}. According to the above runtime estimates \concorde\ would need
for the 1000 vertex instance $T_{296,38}$ more than $10^{27}$ times as long.

\subsection{Comparison of Runtime Results}

In Figure~\ref{fig:RuntimeComparison} we compare the runtimes of all TSPLIB instances with up to 200 vertices with
the runtimes for the $T_{n,m}$  instances with $n$ and $m$ chosen according to~(\ref{eqn:nm}) and the $P_{n,d}$
instances with up to 200 points. As one can see already for quite small
instances the runtimes of \concorde\ on the $T_{n,m}$ instances are by several orders of magnitude larger than on the TSPLIB instances.
Moreover, \concorde's runtime on the $T_{n,m}$ instances is two orders of magnitude slower than for the $P_{n,d}$ instances.
Therefore, the $T_{n,m}$  instances with  $n$ and $m$ chosen according to~(\ref{eqn:nm}) may be useful benchmark instances for
TSP algorithms. All these instances are available for download at~\cite{Tnmurl}.  
It should be noted that there exist polynomial time algorithms for the $P_{n,d}$ instances and the $T_{n,m}$ instances if their special 
structure is exploited~\cite{RTW2001}.

\begin{figure}[ht]
\center
\includegraphics[width=0.9\textwidth]{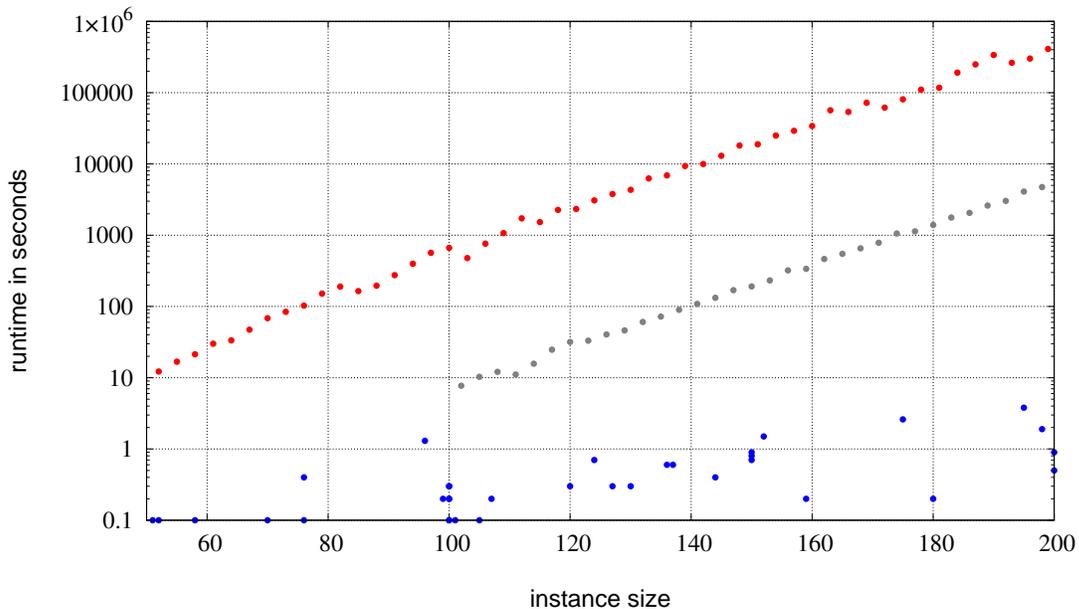}
\caption{\small Comparison of the runtimes of \concorde\ on the $T_{n,m}$ instances (red dots), 
on the $P_{n,d}$ instances (gray dots),
and on the TSPLIB instances (blue dots). The instances have sizes between 51 and at most 200 vertices.
The values $n$ and $m$ for the $T_{n,m}$ instances were chosen according to the equations~(\ref{eqn:nm}).
The x-axis shows the instance size while the y-axis gives the log scaled runtime in seconds. 
The runtime shown for the $T_{n,m}$ instances is the average taken over 10 independent runs of \concorde.
The runtime shown for the $P_{n,d}$ instances and the TSPLIB instances is the average taken over 100 independent runs of \concorde.}
\label{fig:RuntimeComparison}
\end{figure}

\bibliographystyle{plain}

\end{document}